\documentclass[a4paper,12pt]{article}
\usepackage[utf8]{inputenc}
\usepackage{fancybox, calc}
\usepackage{amssymb}
\usepackage{amsmath}
\usepackage{amsthm}
\usepackage{hyperref}
\usepackage{graphicx}
\usepackage{ucs}
\usepackage{amsfonts}
\usepackage[makeroom,thicklines]{cancel}
\usepackage{amssymb}
\usepackage{makeidx}
\usepackage{cancel}
\usepackage{ulem} 
\usepackage{diagbox}
\usepackage{tabularx}
\usepackage{float}
\usepackage[usenames, dvipsnames]{color}
\usepackage{url}
\usepackage{transparent}
\usepackage{enumerate}
\usepackage{eso-pic}
\usepackage{cite} 

\newcommand{\pdf}[0]{f}
\newcommand{\pdfr}[0]{h}

\newcommand{\sign}[0]{\text{sign}}
\def\Rset{\mathbb{R}}

\def\esssup{\operatorname{esssup}}

\newtheorem{theorem}{Theorem}[section]

\newtheorem{lemma}{Lemma}[section]
\newtheorem{definition}{Definition}[section]
\newtheorem{proposition}{Proposition}[section]

\newtheorem{remark}{Remark}[section]

\definecolor{rougeG}{rgb}{.76,0,.12}
\definecolor{vertG}{rgb}{.07,.56,.25}

\def\Rset{\mathbb{R}}
\def\pdf{f}

\numberwithin{equation}{section}

\restylefloat{table}

\title{
Sharp informational inequalities involving Kullback-Leibler and Rényi divergences and a family of scaling-invariant relative Fisher measures
}
\author{ {\bf Razvan Gabriel Iagar\footnotemark[1]\footnote{\textit{e-mail}: razvan.iagar@urjc.es}, David Puertas-Centeno\footnotemark[2]\footnote{\textit{e-mail}: david.puertas@urjc.es}} \\ and {\bf Elio V. Toranzo\footnotemark[3]\footnote{\textit{e-mail}: elio.vtoranzo@urjc.es}} \\ \\
	Departamento de Matem\'atica Aplicada, \\ Ciencia e Ingeniería de los Materiales y Tecnología Electrónica\\
	Universidad Rey Juan Carlos \\
	M\'ostoles, 28933, Madrid, Spain}
\date{\today}

\begin{document}
\maketitle

\begin{abstract}

We introduce a new transformation called \emph{relative differential-escort}, which extends the usual differential-escort transformation by relating the change of variable to a reference probability density. As an application of it, we define a biparametric family of \emph{relative Fisher measures} presenting significant advantages with respect to the pre-existing ones in the literature: invariance under scaling changes and, consequently, sharp inequalities between the new relative Fisher measure and the well established Kullback-Leibler and Rényi divergences. We also introduce a biparametric family of \emph{relative cumulative moment-like measures} and we establish sharp lower bounds of these new measures by the Kullback-Leibler and Rényi divergences. The optimal bound and the minimizing densities are given. We also construct a family of inequalities for an arbitrary and fixed minimizing density in which the so-called generalized trigonometric functions plays a central role, providing thus one more interesting application of the newly introduced inequalities and measures.

\end{abstract}

\section{Introduction}
The construction of informational functionals to grasp different aspects of the spreading of a probability density function beyond its moments has been addressed by several approaches, giving rise to a huge variety of such measures. Starting from the basic ones established in the early days of Information Theory such as Fisher information~\cite{Fisher1922}, Shannon entropy~\cite{Shannon1949,Kolmogorov1956}, Rényi entropy~\cite{Renyi1961}, among other extensions of them \cite{Kapur1969,Taneja1989,Taneja1989(libro),Quesada-Taneja1994,Ilic2021,Lutwak2005,Bercher2012a}, the theory advanced towards numerous generalizations in distinct frameworks, such as cumulative residual entropies~\cite{Rao2004}, group entropies~\cite{Tempesta2019} and $\phi-$entropies~\cite{Zozor2015,Toranzo2017}. Recently, interesting extensions of Jensen- and Fisher-like measures based on mixture of densities has been proposed~\cite{Kharazmi2023}.

The intrinsic mathematical properties of the different informational functionals has held a prominent role in many scientific disciplines from both applied and theoretical standpoints. Alongside the development of communication theory itself, another historic example is the Kennard inequality~\cite{Kennard1927}, which can be considered as the mathematical expression of the Heisenberg uncertainty principle, generalized first by Robertson~\cite{Robertson1929} and later by Schrödinger~\cite{Schrodinger1930}. The incursion of Information Theory into this topic from Quantum Mechanics gave rise to the entropic formulation of the uncertainty principle using the Shannon entropy \cite{Bialynicki1975,Bialynicki1984}, which paved the way to generalized versions involving Rényi entropies~\cite{Bialynicki2006} as well as other informational measures, such as the Fisher information ~\cite{Romera2006}. Beyond these uncertainty inequalities, the interconnections among the different informational functionals has given rise to a rich field of research. In this direction, the authors of~\cite{Toranzo2017} proposed the notion of \textit{informational triangle} to refer to the triad  of moments, entropies, and Fisher information measures as well as the links between each pair of them~\cite{Dembo1991}. Relevant examples of that are the informational inequalities connecting moments, entropies and Fisher information measures, such as Stam, Cramér-Rao and moment-entropy inequalities, which establish lower bounds to products of those functionals. All of them are minimized by the Gaussian density, as with the Kennard and Biaylinicky-Birula inequalities. Extensions of these families of inequalities have been studied in the last years~\cite{Lutwak2005,Bercher2012a} involving parametric versions of the classical Fisher information as well as the Rényi entropies, which are minimized by the stretched Gaussian densities, ubiquitous in nonextensive systems~\cite{Tsallis2022b}.

Within the same framework, a parallel task is devoted to quantify the dissimilarity between two (or more) probability density functions by constructing relative versions of the functionals associated to the Shannon and generalized entropies, as well as the Fisher information measure. Nevertheless, up to now this task remains a significant challenge in information theory~\cite{Sason2022}. Among the relative functionals already known in the literature, Kullback-Leibler~\cite{Kullback1951,Lin2002} and Rényi divergences~\cite{Hammad1978} are the most representative and employed measures in a wide variety of areas of science, see for example \cite{Singh1997,Karmeshu2003,Singh2013,Sen2011,Sen2014}. In certain sense, they can be considered as the relative counterparts of Shannon and Rényi entropies. While the mathematical properties of these measures of divergence have been thoroughly studied~\cite{Harremoes2014}, as far as we know, only a few proposals of relative Fisher-like measures can be found in the literature as the counterpart of Fisher-like measures in the \textit{relative framework} \cite{Hammad1978,Antolin2009,Martin2015,Toscani2016(a),Toranzo2017,Toscani2017}. Furthermore, none of them satisfy the property of invariance under scaling transformations, or play a role in informational inequalities involving the Kullback-Leibler or Rényi divergences. We think that a great goal of this work is to address this serious lack by the introduction of a biparametric family of relative Fisher measures which does fulfill the above mentioned properties, and allow to bound Kullback-Leibler and Rényi divergences.

Furthermore, with the aim to reconstruct the informational triangle in the relative framework, we also introduce a biparametric family of functionals which turns out to be a relative version of the recently introduced cumulative moments~\cite{Puertas2025}.We have called them \textit{relative cumulative moments}. Indeed, we  show that they also fulfill the same relevant properties than the here introduced relative Fisher measures, that is, scaling invariance and boundedness by Rényi and Kullback-Leibler divergences. In fact, we manage to transfer the generalized moment-entropy and Stam inequalities derived by Lutwak~\cite{Lutwak2005} and Bercher~\cite{Bercher2012a} to the relative framework. Moreover, on the one hand we establish the optimal bound and the minimizing densities of those inequalities with respect to a reference probability density. On the other hand, we transfer the monotonicity property of the classical Stam product~\cite{Rudnicki2016} to its relative counterpart, namely, we prove the existence of a class of transformations to which the mentioned product is monotone~\cite{Rudnicki2016,Puertas2025}.

In order to tackle the previously mentioned properties we apply alike techniques to the ones used in~\cite{Zozor2017,Puertas2019,Puertas2025} as well as in~\cite{IP2025,IP2025(b)}, which involve the use of transformations the \textit{differential-escort} and \textit{up/down} transformations, respectively. More precisely, we introduce here a new transformation which can be interpreted as the relative counterpart of the differential-escort transformation. Its definition is given with respect to a fixed probability density, where the standard differential-escort case is recovered only when the reference density corresponds to the uniform. Differential-escort and up/down transformations have been shown to be powerful tools in the study of informational inequalities, presenting several advantages with respect to the standard escort transformations~\cite{Tsallis2009}. While differential-escort transformations have allowed to further extend generalized Stam and moment-entropy inequalities~\cite{Zozor2017,Puertas2025}, up/down transformations directly connect moments, entropies and Fisher measures among them (in the sense of~\cite[Lemma 3.1]{IP2025}), but also reveal a mirrored domain of the entropic parameters where the inequalities hold true, in which the minimizing densities are unbounded in the extremes. Even more, these transformations have motivated the definition of new classes of informational functionals which, in turn, allow to establish different upper bounds for the product of the main informational functionals depending on regularity conditions on the probability density~\cite{IP2025(b)}.

Finally, we use the fact that the minimizing densities of the former inequalities depend on an arbitrary probability density to impose a fixed minimizer. Astonishingly, this gives rise to a modified class of relative functionals, which we have named \textit{adapted relative measures}, which involves the so-called \textit{generalized trigonometric functions} in an intricate way. We recall that the generalized trigonometric functions are special functions that have begun to receive increased attention only some decades ago~\cite{Lindqvist1995,Drabek1999} due their role in an eigenvalue problem associated with the $p$-Laplacian operator. Remarkably, they also play a central role in the extension of moment-entropy, Cramér-Rao and Stam inequalities~\cite{Puertas2025,IP2025}. But beyond that, other interesting applications of these functions have been found in theoretical physics~\cite{Shababi2016}, to model the vibration properties of the human teeth~\cite{Cveticanin2020a}, or to express exact solutions of a nonlinear Schrödinger equation~\cite{Gordoa2025}, which makes them a research hotspot.

Returning to the adapted relative measures, it is worth mentioning that, although these new functionals depend on two probability densities, they must not be understood as divergence measures since they are not minimized when both densities are the same, but instead the corresponding product of both functionals does. An outstanding characteristic of the inequalities involving the product of these adapted measures is the fact that, while each element of the inequality is minimized by the same density, the combination of both elements is minimized by the desired density. This behavior is in a stark contrast with the rest of informational inequalities in which each component is minimized by different densities as it is deeper analyzed throughout the paper.

\bigskip

\noindent \textbf{Structure of the paper}. A number of preliminary definitions and results needed in the development of the paper are recalled, for the sake of completeness, to the reader in Section \ref{sec:prelim}. In Section \ref{sec:transf} we introduce the relative differential-escort transformation and its inverse transformation, we study their basic properties, as well as their relation with the Shannon and Rényi entropies and with the Kullback-Leibler and Rényi divergences. Section \ref{sec:infineq} is dedicated to the definition of the biparametric families of relative Fisher measures and relative cumulative moments. We therein explore their basic properties allowing us later to establish the new informational inequalities of moment-entropy and Stam-like types involving the new informational functionals together with the Kullback-Leibler and Rényi divergences. We furthermore extend the monotonicity property (in the sense of~\cite{Rudnicki2016}) of the Fisher-Shannon complexity measure to the relative framework. In Section \ref{sec:inverse}, we tackle the inverse problem in the sense of finding the reference function for a minimizing density, fixed \emph{a priori}. We then show the unexpected connection with the generalized trigonometric functions, which play a fundamental role in the definition of the adapted informational measures. Finally, in Section \ref{sec:conclusion} we give some conclusions and state some open problems to be addressed in the future.

\section{Preliminary notions}\label{sec:prelim}

Throughout this paper, we consider as a general framework two probability density functions $f$ and $h$ such that
\begin{equation}\label{cond:support}
{\rm supp}\,f={\rm supp}\,h=\overline{\Omega}, \quad \Omega=(x_i,x_f)\subseteq\Rset, \quad f(x)>0, \ h(x)>0, \quad {\rm for \ any} \ x\in\Omega,
\end{equation}
where $\Omega$ can be either bounded or unbounded and $\overline{\Omega}$ denotes the closure of the set $\Omega$. Moreover, for $r\in(0,\infty)$, we denote by $\mathcal{S}_r$ the scaling transformation
\begin{equation}\label{eq:scaling}
\mathcal{S}_r[f](x):=rf(rx),
\end{equation}
and we will use the abbreviated notation $f_{(r)}=\mathcal{S}_r[f]$. The following expression will also appear frequently in the sequel:
\begin{equation}\label{eq:xi}
\xi(\lambda,\alpha):=1+\alpha(\lambda-1).
\end{equation}
In the rest of this section, for the sake of completeness, we review the notions, concepts and previously established results that are useful in this work.

\subsection{A review of moments, entropies and classical inequalities}

In this first part, we recall the reader the notions of moments, entropies and Fisher information and the most well known inequalities between them, with the corresponding minimizers.

\medskip

\noindent \textbf{The $p$-th absolute moment} of a probability density function $\pdf$, with $p \geqslant 0$, is given by
\begin{equation}\label{eq:pabs}
\mu_p[\pdf] = \int_\Rset |x|^p \, \pdf(x) \, dx  = \left\langle  \, |x|^p \, \right\rangle_\pdf.
\end{equation}
Related to the $p$-th absolute moment, it is more common to consider the related quantity known as the $p$-th deviation, defined by taking the power $1/p$ in~\eqref{eq:pabs} (including the limiting cases $p\to0$ and $p\to \infty$) as follows:
\begin{equation*}
\begin{split}
&\sigma_p[\pdf] =\left(\int_\Rset |x|^p \, \pdf(x) \, dx \right)^{\frac{1}{p}}, \quad {\rm for} \ p > 0, \\
&\sigma_0[\pdf] = \lim\limits_{p \to 0} \sigma_p[\pdf]  = \exp\left(\int_\Rset \pdf(x) \, \log|x| \, dx \right), \\
&\sigma_{\infty}[\pdf]=\lim\limits_{p\to\infty}\sigma_p[\pdf]=\esssup\big\{|x| : x\in\Rset, \pdf(x) > 0 \big\}.
\end{split}
\end{equation*}

\medskip

\noindent \textbf{Rényi and Tsallis entropies.} The Rényi and Tsallis entropies of $\lambda-$order, $\lambda\neq1$, of a probability density function $\pdf$ defined on $\Rset$ (or on a Lebesgue measurable subset of $\Rset$) are given by
\begin{equation*}
R_\lambda[\pdf] = \frac1{1-\lambda} \log\left( \int_\Rset [\pdf(x)]^\lambda \, dx \right), \quad {\rm and} \quad T_\lambda[\pdf] = \frac1{\lambda-1} \left( 1 - \int_\Rset [\pdf(x)]^\lambda \, dx \right),
\end{equation*}
respectively. Note that the Rényi and Tsallis entropies are one-to-one mapped as follows:
$$
T_{\lambda}[\pdf]=\frac{e^{(1-\lambda) R_{\lambda}[\pdf]} - 1}{1 - \lambda}.
$$
In the limit as $\lambda\to1$ we recover the Shannon entropy
\begin{equation*}
\lim\limits_{\lambda \to 1}R_{\lambda}[\pdf]=\lim\limits_{\lambda \to 1} T_{\lambda}[\pdf]=S[\pdf]=-\int_\Rset \pdf(x) \, \log\pdf(x) \, dx.
\end{equation*}
Throughout this paper, we employ the quantity known as the Rényi entropy power, that is,
$$
N_\lambda[\pdf] = e^{R_{\lambda}[\pdf]} = \left\langle\pdf^{\lambda-1}(x)\right\rangle^\frac1{1-\lambda}_\pdf.
$$
We will also denote by $N[\pdf]=e^{S[\pdf]}$ the Shannon entropy power.

\medskip

\noindent \textbf{$(p,\lambda)-$Fisher information.} The $(p,\lambda)-$Fisher information is a generalization of the usual Fisher information and it was introduced by Lutwak and Bercher~\cite{Lutwak2005, Bercher2012, Bercher2012a}. It applies to derivable probability density functions. More precisely, given $p>1$ and $\lambda \in \Rset\setminus\{0\}$, the $(p,\lambda)-$Fisher information of a probability density function $f$ is defined as
\begin{equation}\label{eq:def_FI}
\phi_{p,\lambda}[\pdf]
\: = \: \left(\int_\Rset \left|\pdf(x)^{\lambda-2} \, \frac{d\pdf}{dx}(x)\right|^{p} \pdf(x) \, dx\right)^{\frac1{p\lambda}}
\: = \: \frac1{|\lambda-1|^{\frac1\lambda}} \, \left\langle \left|\frac {d \pdf^{\lambda-1}}{dx}(x)\right|^p \right\rangle_{\!\pdf}^\frac1{p \lambda},
\end{equation}
when $\pdf$ is differentiable on its support. Observe that the $(2,1)-$Fisher information reduces to the standard Fisher information. Another interesting particular case is the $(1,\lambda)-$Fisher information, corresponding to the total variation of $\frac{\pdf^\lambda}{\lambda}$, provided that $\pdf^\lambda$ is a function with bounded variation.

\medskip

\noindent \textbf{The extended (triparametric) Stam inequality}. This is a generalized inequality, established with three parameters in~\cite{Zozor2017} and extended by two of the authors in their previous work \cite[Theorem 5.1]{IP2025} with the help of the differential-escort transformation. Let $p\geqslant 1$ and $\beta$ be such that
\begin{equation}\label{eq:sign_cond1}
\sign\left(p^*\beta + \lambda-1\right)=\sign\left(\beta+1-\lambda\right)\neq0,
\end{equation}
where $p^*$ designs the Hölder conjugate of $p$. Then, the following generalized Stam inequality holds true for $f: \Rset\mapsto\Rset^+$ absolutely continuous if $1+\beta-\lambda>0$ or for $f:(x_i,x_f)\mapsto \Rset^+$ absolutely continuous on $(x_i,x_f)$ if $1+\beta-\lambda<0$:
\begin{equation}\label{ineq:trip_Stam_extended}
\left(\phi_{p,\beta}[\pdf] \, N_{\lambda}[\pdf] \right)^{\theta(\beta,\lambda)}\: \geqslant \:	\left(\phi_{p,\beta}[g_{p,\beta,\lambda}] \, N_{\lambda}[g_{p,\beta,\lambda}] \right)^{\theta(\beta,\lambda)}\equiv K_{p,\beta,\lambda}^{(1)},
\end{equation}
where $\theta(\beta,\lambda)=1+\beta-\lambda$ and $g_{p,\beta,\lambda}$ is defined in \cite[Section 4]{Puertas2025}\footnote{Throughout the manuscript, the notation used for the lower bound corresponding to the Stam and moment-entropy inequalities, respectively, follows the one introduced in \cite{IP2025}.}. It should be noted that this result has been extended even further to a \textit{mirrored} domain of the parameters~\cite{IP2025}. However, for simplicity and in order to highlight the main results of the work, we refrain here to consider the mirrored domain.

\medskip

\noindent \textbf{The moment-entropy inequality} is an informational inequality relating the R\'enyi entropy power and the $p$-th deviation $\sigma_p$. More precisely, when
\begin{equation}\label{eq:param_clas}
p^*\in[0,\infty ), \quad \lambda>\frac1{1+p^*},
\end{equation}
and $f$ is any probability density function, it was established in~\cite{Lutwak2004, Lutwak2005, Bercher2012} that
\begin{equation}\label{ineq:bip_E-M}
\frac{\sigma_{p^*}[\pdf]}{N_{\lambda}[\pdf]} \, \geqslant \, \frac{\sigma_ {p^*}[g_{p,\lambda}]}{N_{\lambda}[g_{p,\lambda}]}\equiv K^{(0)}_{p,\lambda}.
\end{equation}
The minimizers to Eq. \eqref{ineq:bip_E-M} are the stretched Gaussians $g_{p,\lambda}$,  given by
\begin{equation}\label{def:g_plambda}
g_{p,\lambda}(x) \, = \, \frac{a_{p,\lambda}}{\exp_\lambda\left( |x|^{p^*} \right)}
\, = \, a_{p,\lambda}\, \exp_{2-\lambda}\left( - |x|^{p^*} \right),
\end{equation}
where $\exp_\lambda$ is the generalized Tsallis exponential
\begin{equation}\label{def:q-exp}
\exp_\lambda(x) = \left( 1 + (1 - \lambda) \, x \right)_+^\frac1{1-\lambda}, \ \ \lambda \ne 1, \qquad \exp_1(x) \: \equiv \: \lim_{\lambda \to 1} \, \exp_\lambda(x) \: = \: \exp(x),
\end{equation}
and the normalization constants $a_{p,\lambda}$ have a rather complicated explicit value which we omit here (see \cite{IP2025}).

\subsection{A brief recall of R\'enyi divergences}

This second preliminary part is dedicated to a short presentation of the more specific objects related to the theme of this paper, the divergences, functionals designed to measure the ``distance" or dissimilarity between two probability density functions. We thus recall next the definitions of the R\'enyi and Kullback-Leibler divergences of two density functions, following \cite{Harremoes2014}.
\begin{definition}[$\xi$-R\'enyi divergence]
Let $f$ and $h$ be two probability density functions satisfying \eqref{cond:support} and $\xi>0$ be a positive real number. Then, the $\xi$-Rényi divergence (or mutual information) is defined as
	\begin{equation}\label{def:renyi_mutual}
	D_\xi[\pdf||h]=\frac1{\xi-1}\log \int_{\Omega} [\pdf(t)]^{\xi}[h(t)]^{1-\xi}\,dt,\quad \xi\neq 1.
	\end{equation}
In the limit $\xi\to1$ we recover the Kullback-Leibler divergence defined as follows
	\begin{equation}\label{def:kullback-leibler}
	D[f||h]:= \int_{\Omega}  f(t)\log\left[\frac{f(t)}{h(t)}\right]\,dt.
	\end{equation}
\end{definition}
Throughout the paper we also employ the exponential of the R\'enyi divergence and thus adopt the following notation
\begin{equation}\label{eq:defK}
K_{\xi}[f||h]:=e^{(\xi-1)D_\xi[f||h]}=\int_{\Omega} [\pdf(t)]^{\xi}[h(t)]^{1-\xi}\,dt,\quad \xi\neq1,
\end{equation}
which is related to the Chernoff $\alpha$ distance, also called the skewed Bhattacharyya distance, see for example \cite{Chernoff1952,Yamano2021}.
It is obvious that $K_0[f||h]=1$ for any pair of densities $f$, $h$, and if $\xi\neq0$, then $D_\xi[f||h]=0$ if and only if $f=h$. Moreover, as both $f$ and $h$ are taken to be positive in $\Omega$, one can extend the definition of $K_{\xi}$ to any real number $\xi$, if the integral in the right hand side of \eqref{eq:defK} is finite.

\begin{remark}[Positivity of the Rényi divergence]
Given any pair of probability densities $f$, $h$ satisfying \eqref{cond:support} and $\xi>0$, the Rényi and Kullback-Leibler divergences are positive, that is,
$$
D_\xi[f||h]> 0,
$$
and the equality is achieved if and only if $f=h$. This implies that
$$
K_{\xi}[f||h]>1 \quad {\rm for \ any} \ \xi>1, \qquad  K_{\xi}[f||h]\in(0,1) \quad {\rm for \ any} \ \xi\in(0,1).
$$
\end{remark}
It is worth noticing that the Rényi and Kullback-Leibler divergences are invariant under scaling transformations acting on both probability densities $f$ and $g$, that is,
$$
D_\xi[f_{(r)}||g_{(r)}]=D_\xi[f||g], \quad r>0.
$$ 	

\subsection{Generalized trigonometric functions}\label{sec:GTFs}

We end the section devoted to a recall of the notions and results employed throughout the paper by a brief recap of the generalized trigonometric functions. Although historically the generalized trigonometric functions, as well as their hyperbolic counterparts, have been first introduced with a single parameter by Shelupsky~\cite{Shelupsky1959}, and Lindqvist~\cite{Lindqvist1995}, the $(p,q)$-generalized trigonometric functions $\sin_{p,q},\cos_{p,q},\tan_{p,q}$ involved in this work and their respective hyperbolic counterparts were first defined, as far as we know, by Drábek and Manásevich who related them with an eigenvalue problem involving the p-Laplacian \cite{Drabek1999}. The $(p,q)$-sine function is defined as the inverse of
\begin{equation} \label{arcsin_def}
\text{arcsin}_{p,q} (z)=\int_0^z(1-t^q)^{-\frac1p}dt=z\,_2F_1\left(\frac 1p,\frac1q;1+\frac 1q;z^q\right),  \quad p,q>1,
\end{equation}
which, as we see, can also be expressed using the Gaussian hypergeometric function (see for example \cite{Edmunds2012, Bhayo2012}). The cosine function is defined as the derivative of the sine function
\begin{equation*}
\cos_{p,q}(z)=\frac{d}{dz}\sin_{p,q}(z).
\end{equation*}
From these definitions, and applying the inverse function rule, the following Pythagorean-like identity is obtained:
\begin{equation}\label{eq:pyth}
\cos_{p,q}^p(z)+\sin_{p,q}^q(z)=1.
\end{equation}
In the hyperbolic case one has
\begin{equation}\label{arcsinh_def}
\text{arcsinh}_{p,q} (z)=\int_0^z(1+t^q)^{-\frac1p}dt=z\,_2F_1\left(\frac 1p,\frac1q;1+\frac 1q;-z^q\right),
\end{equation}
and it follows that
\begin{equation}\label{eq:pythh}
\cosh^p_{p,q}(z)-\sinh^q_{p,q}(z)=1,
\end{equation}
where
\begin{equation*}
\cosh_{p,q}(z)=(\sinh_{p,q}(z))'.
\end{equation*}
A more detailed description of the properties of the generalized trigonometric functions and more applications of them can be found in very recent papers such as \cite{Gordoa2025, Puertas2025}.

\section{The relative differential-escort transformation}\label{sec:transf}

In this section we introduce the main object of study in the present work, which is a new mapping between probability density functions that we call \emph{the relative differential-escort transformation}, as it is an extension to the relative framework of the standard differential-escort transformation. We also construct its inverse and study its basic properties.

\subsection{Main notions}

The relative differential-escort transformation is defined as follows:
\begin{definition}\label{def:transf}
Let $f$ and $h$ be two probability density functions satisfying \eqref{cond:support} and $\alpha\in\Rset$. We define the \textit{relative differential-escort} transformed density of $\alpha$-order of $f$ as
\begin{equation}\label{eq:transf}
\mathfrak R_\alpha^{[\pdfr]}[\pdf](y):=\left(\frac{\pdf(x(y))}{\pdfr(x(y))}\right)^\alpha,\quad y'(x)=\pdf(x)^{1-\alpha}\pdfr(x)^{\alpha}.
\end{equation}
\end{definition}
For simplicity, we also employ throughout the paper the alternative notation
$$
f_\alpha^{[h]}(y)\equiv \mathfrak R_\alpha^{[\pdfr]}[\pdf](y).
$$
Observe first that $\mathfrak R_\alpha^{[\pdfr]}[\pdf]$ is a probability density function. Indeed,
$$
\int_{\Rset}f_\alpha^{[h]}(y)\,dy=\int_{\Rset}\frac{h^{\alpha}(x)}{f^{\alpha}(x)}f^{1+\alpha}(x)h^{-\alpha}(x)\,dx=\int_{\Rset}f(x)\,dx=1.
$$
Notice that, if we let $h\equiv1$ in $\Omega$, we are left with the differential-escort transformation of $\alpha$-order, denoted by $\mathfrak E_\alpha$, considered in previous works such as \cite{Puertas2017, Puertas2019, Zozor2017}. This particular case motivates the name given to this generalized transformation depending on a reference function. Let us remark here that the composition of the $\beta$-order differential-escort transformation with the $\alpha$-order relative differential-escort gives the $\alpha\beta$-order relative differential-escort transformation, more precisely:
\begin{equation}\label{eq:comp}
\mathfrak E_{\beta}\mathfrak R_\alpha^{[h]}[f]=\left(f_\alpha^{[h]}\right)^\beta=\left(\frac{f}{h}\right)^{\alpha\beta}=\mathfrak R_{\alpha\beta}^{[h]}[f].
\end{equation}
We also notice that the previous definition depends on an integration of the new independent variable, thus it is defined up to a translation with an integration constant. Assuming, for example, as a ``canonical election" that
$$
y(x)=\int_{x_i}^x f^{1-\alpha}(t)h^{\alpha}(t)\,dt,
$$
we readily get that $y(x_i)=0$, $y(x_f)=y_f>0$, hence the support of the transformed density $f_\alpha^{[h]}(y)$ is the interval $(0,y_f)$ (where $y_f$ might be finite or infinite, depending on whether the previous integral is convergent or divergent as $x\to x_f$). However, we shall denote by $\widetilde{\Omega}_{\alpha}$ the domain of the transformed density $f_\alpha^{[h]}$ in the $y$ variable.
\begin{remark}
Let us observe that the support $\widetilde\Omega_\alpha$ is finite if and only if $K_\alpha[h||f]$ is finite. In fact,
\begin{equation}
L(\widetilde{\Omega}_\alpha)=K_\alpha[h||f].
\end{equation}

\end{remark}
The next result deals with the derivative of the transformed density $f_\alpha^{[h]}$.
\begin{lemma}\label{lem:deriv}
If the densities $f$ and $h$ satisfying \eqref{cond:support} are both derivable, then $f_\alpha^{[h]}$ is derivable as well and
\begin{equation}\label{eq:deriv_transf}
\frac{df_\alpha^{[h]}}{dy}(y(x))=\frac{\alpha f^{2\alpha-1}(x)}{h^{2\alpha}(x)}\left(\frac{f'(x)}{f(x)}-\frac{h'(x)}{h(x)}\right).
\end{equation}
\end{lemma}
\begin{proof}
We proceed by direct calculation:
\begin{equation*}
\begin{split}
\frac{df_\alpha^{[h]}}{dy}(y(x))&=\alpha\left(\frac{f(x)}{h(x)}\right)^{\alpha-1}\frac{d}{dx}\left(\frac{f}{h}\right)x'(y)\\
&=\alpha\left(\frac{f(x)}{h(x)}\right)^{\alpha-1}\frac{f'(x)h(x)-h'(x)f(x)}{h^2(x)}\frac{f^{\alpha-1}(x)}{h^{\alpha}(x)}\\
&=\frac{\alpha f^{2\alpha-1}(x)}{h^{2\alpha}(x)}\left(\frac{f'(x)}{f(x)}-\frac{h'(x)}{h(x)}\right),
\end{split}
\end{equation*}
as claimed.
\end{proof}

\bigskip

\noindent \textbf{Construction of the inverse transformation.} Another basic property that a well defined transformation should have is an inverse, as one would like to be able to recover the initial information from the result obtained by applying the transformation. In order to define the inverse, let us first observe that, given any probability density function $g$ and $\gamma\in\Rset\setminus\{0,1\}$ with the property that $N_{\gamma}[g]<\infty$, there is a unique scaling parameter $r=N_{\gamma}[g]$ such that
$$
\int_{\Rset}\mathcal{S}_{r}[g]^{\gamma}(x)\,dx=1.
$$
Indeed, a direct change of variable gives
$$
\int_{\Rset}\mathcal{S}_{r}[g]^{\gamma}(y)\,dy=\int_{\Rset}r^{\gamma}g^{\gamma}(rx)\,dx=r^{\gamma-1}\int_{\Rset}g^{\gamma}(x)\,dx,
$$
whence we obtain that
\begin{equation}\label{eq:interm3}
r=\left[\int_{\Rset}g^{\gamma}(x)\,dx\right]^{\frac{1}{1-\gamma}}=N_{\gamma}[g].
\end{equation}
With this in mind, we introduce the following definition:
\begin{definition}\label{def:inverse}
Given a probability density $h$ as reference function, the inverse of the relative differential-escort transformation is defined as follows: for any $\alpha\in\Rset\setminus\{1\}$ and any probability density function $g$ such that
\begin{equation}\label{cond:inverse}
r(\alpha,g):=N_{\frac{\alpha-1}{\alpha}}[g]=N_{\frac{1}{\alpha^*}}[g]<\infty,
\end{equation}
we define
\begin{equation}\label{eq:inversegen}
\mathfrak{R}^{-1,[h]}_{\alpha}[g](x)=\mathcal{S}_{r(\alpha,g)}[g](y(x))^{\frac{1}{\alpha}}h(x), \quad y'(x)=\mathcal{S}_{r(\alpha,g)}[g](y(x))^{-\frac{1}{\alpha^*}}h(x),
\end{equation}
which can be equivalently written as
\begin{equation}\label{eq:inversegen2}
\mathfrak{R}^{-1,[h]}_{\alpha}[g](x)=[rg(ry(x))]^{\frac{1}{\alpha}}h(x), \quad y'(x)=[rg(ry(x))]^{-\frac{1}{\alpha^*}}h(x),
\end{equation}
with $r=r(\alpha,g)$. For $\alpha=1$ and any compactly supported probability density function $g$, we define
\begin{equation}\label{eq:inverse}
\mathfrak{R}^{-1,[h]}_{1}[g](x)=\mathcal{S}_{r(1,g)}[g](y(x))h(x), \quad y'(x)=h(x),
\end{equation}
where $r(1,g)$ is the length of the support of $g$.
\end{definition}
Note that the implicit change of the independent variable in \eqref{eq:inversegen} can be equivalently written as
$$
\int \mathcal{S}_{r(\alpha,g)}[g]^{\frac{1}{\alpha^*}}\,dy=\int h(x)\,dx,
$$
which entails in particular that the inverse is well defined only for functions $g$ such that
$$
\int_{\widetilde{\Omega}}g(y)^{\frac{1}{\alpha^*}}\,dy<\infty
$$

Note also that the election of the scaling parameter $r(\alpha,g)$ implies that
$$
\int_\Rset  \mathcal{S}_{r(\alpha,g)}[g]^{\frac{\alpha-1}{\alpha}}\,dy=\int_\Rset h(x)dx=1.
$$
The latter discussion justifies the condition of finiteness of the R\'enyi entropy given in \eqref{cond:inverse}. A similar argument justifies that the inverse for $\alpha=1$ is well defined only for compactly supported densities. We infer from \eqref{eq:comp} that
$$
\mathfrak{R}^{-1,[h]}_{\alpha\gamma}\mathfrak{E}_{\gamma}\mathfrak{R}^{[h]}_{\alpha}=\mathcal{I},
$$
and we deduce after straightforward manipulations that
\begin{equation}\label{eq:comp_inverse}
\mathfrak{R}^{-1,[h]}_{\beta}\mathfrak{E}_{\gamma}=\mathfrak{R}^{-1,[h]}_{\frac{\beta}{\gamma}}.
\end{equation}
In particular, letting $\beta=1$ and renaming $\gamma$ instead of $1/\gamma$ in \eqref{eq:comp_inverse}, we arrive at
\begin{equation}\label{eq:comp_inverse1}
\mathfrak{R}^{-1,[h]}_{\gamma}=\mathfrak{R}^{-1,[h]}_{1}\mathfrak{E}_{\frac{1}{\gamma}}, \quad {\rm for \ any} \ \gamma>0.
\end{equation}
We believe that the scaling change taking as scaling parameter exactly the R\'enyi entropy power of the density (provided it is finite) in order to define the inverse is a very interesting and unexpected feature of the relative differential-escort transformation.

We highlight below the fact that a slight modification of the transformation might be employed in order to force that the length of the support is kept invariant.
\begin{remark}
Whenever the relative differential-escort transformation preserves the compact or non-compact character of the support of the original probability density function, the transformed function can be scaled in order to keep invariant its support. For example, if the support $\Omega$ of both $f$ and $h$ is a compact set, easy manipulations show that the previous fact is achieved by performing a scaling change of order
$$
\mathcal A=\frac{D_\alpha[h||f]}{L(\Omega)},
$$
provided that $D_\alpha[h||f]<\infty$, where $L(\Omega)$ is the length of the common support $\Omega$. Thus, the following transformation
$$
\overline {\mathfrak R}_\alpha^{[h]}=\mathcal A\overline {\mathfrak R}_\alpha^{[h]}=\mathcal A\left(\frac{f(x)}{h(x)}\right)^\alpha
$$
defines a rearrangement (in the sense of preserving the support and the value of the integral of the initial function on every subset of the support) for each value of $\alpha$ such that $D_\alpha[h||f]<\infty$ if the support of the probability density $h$ is compact. A similar construction, whose details we omit here, takes place when $\Omega$ is not compact and $D_\alpha[h||f]=\infty$.
\end{remark}

\subsection{An easy example: exponential densities}

In order to show how the relative differential-escort transformation works, let us give a practical example by considering a density $f(x)=ae^{-ax}$ for some $a\in\Rset^+$ and $h(x)=e^{-x}$, both of them supported on $\Omega=(0,\infty)$. We infer then from \eqref{eq:transf} that
$$
f^{[h]}_\alpha(y)=\left(\frac{ae^{-ax}}{e^{-x}}\right)^\alpha=a^\alpha e^{-\alpha (a-1)x},
$$
with
$$
y'(x)=f(x)^{1-\alpha}h(x)^\alpha=a^{1-\alpha}e^{-(1-\alpha)ax}e^{-\alpha x}=\frac {e^{-(a+\alpha-a\alpha)x}}{a^{\alpha-1}}.
$$
Then
$$
y(x)=\int_0^x \frac {e^{-(a+\alpha-a\alpha)t}}{a^{\alpha-1}}dt=\frac {1-e^{-(a+\alpha-a\alpha)x}}{a^{\alpha-1}(a+\alpha-a\alpha)},
$$
and we obtain
$$
e^{-(a+\alpha-a\alpha)x}=a^{\alpha-1}(a\alpha-a-\alpha)\,y+1
=a^{\alpha}\left(\alpha-1-\frac{\alpha}a\right)\,y+1=a^{\alpha}\left(\frac{\alpha}{a^*}-1\right)\,y+1.
$$
We can thus replace the previous expression and continue the calculation of the transformed density:
$$
f^{[h]}_\alpha(y)=a^\alpha \left(e^{-(a+\alpha-a\alpha)x}\right)^{\frac{\alpha (a-1)}{a+\alpha-a\alpha}}
=a^\alpha\left[a^{\alpha}\left(\frac{\alpha}{a^*}-1\right)\,y+1\right]^{\frac{\alpha (a-1)}{a+\alpha-a\alpha}}.
$$
Note that, if the following condition
$$
a+\alpha-a\alpha=a\left(1-\frac\alpha{a^*}\right)>0,\quad {\rm or \ equivalently,}\quad \alpha<a^*,
$$
is fulfilled, then the support of $f_\alpha^{[h]}(y)$ is compact. On the contrary, when $\alpha>a^*$ the support is infinite and with a power-law tail. Finally, in the case $\alpha=a^*$, we obtain again an exponential density by taking into account that the change of variable reduces to a trivial scaling.

\subsection{Entropies and divergences}

We give here an interesting connection between the R\'enyi entropy power and the R\'enyi divergence, achieved through the relative differential-escort transformation. In particular, a quite surprising connection between the Shannon entropy and the Kullback-Leibler divergence is established.
\begin{lemma}\label{lem:Renyi}
	Let $f$ and $h$ be two probability density functions satisfying \eqref{cond:support}, $\alpha\in\Rset$ and $\lambda\in\Rset\setminus\{1\}$. Then
	\begin{equation}\label{eq:Renyi}
	N_{\lambda}^{1-\lambda}[f_\alpha^{[h]}]=K_{\xi(\lambda,\alpha)}[f||h], \quad \xi(\lambda,\alpha)=1+\alpha(\lambda-1).
	\end{equation}
	For the Shannon entropy, we find
	\begin{equation}\label{eq:Shannon}
	S[f_\alpha^{[h]}]=-\alpha D[f||h].
	\end{equation}
\end{lemma}
\begin{proof}
	Recalling the definitions in \eqref{eq:transf} and \eqref{eq:defK}, we have
	\begin{equation*}
	\begin{split}
	N_{\lambda}^{1-\lambda}[f_\alpha^{[h]}]&=\int_{\widetilde{\Omega}}(f_\alpha^{[h]})^{\lambda}(y)\,dy
	=\int_{\Omega}\left(\frac{f(x)}{h(x)}\right)^{\alpha(\lambda-1)}f(x)\,dx\\
	&=\int_{\Omega}f(x)^{1+\alpha(\lambda-1)}h(x)^{\alpha(1-\lambda)}\,dx=K_{\xi(\lambda,\alpha)}[f||h],
	\end{split}
	\end{equation*}
	leading to \eqref{eq:Renyi}. Performing an analogous calculation with the Shannon entropy, we obtain
	\begin{equation*}
	\begin{split}
	S[f_\alpha^{[h]}]&=-\int_{\widetilde{\Omega}}f_\alpha^{[h]}(y)\log f_\alpha^{[h]}(y)\,dy=-\int_{\Omega}f(x)\log\left(\frac{f(x)}{h(x)}\right)^{\alpha}\,dx\\
	&=-\alpha\int_{\Omega}f(x)\log\frac{f(x)}{h(x)}\,dx=-\alpha D[f||h]
	\end{split}
	\end{equation*}
	and the proof is complete.
\end{proof}

\medskip

\noindent

\textbf{Entropic relative moment problem.} The previous result has a remarkable consequence related to the moment problem. Let $\{N_i\}_{i}$ be an entropy sequence in the sense of~\cite{IP2025}, that is, there exists a decreasing probability density function $f$ such that $N_i[f]=N_i$. Taking $h$ fixed with the same support as $f$, then the set of Rényi divergences $\{D_i[f||h]\}_i$ also characterizes the probability density $f$, as follows from the identity $D_{i}[f||h]=N_i[\mathfrak R_1^{[h]}[f]],$ and the fact that the transformation $\mathfrak R^h_1$ is invertible.

\medskip

With the definitions and basic properties given in this section in mind, we are in a position to introduce several new informational functionals and establish sharp informational inequalities involving them.

\section{New functionals and informational inequalities}\label{sec:infineq}

This section is dedicated to the statement and proof of new informational inequalities derived with the aid of the relative differential-escort transformation. In order to state the inequalities, we need to define first some informational functionals that will play an important role.

\subsection{Informational functionals}

Motivated by the relative differential-escort transformation, we introduce two new functionals generalizing the Fisher information and the cumulative moment. The first of them is defined below.
\begin{definition}\label{def:relativeFD}
Let $f$, $h$ be two probability density functions satisfying \eqref{cond:support}, being both derivable on their support, and let $(p,\lambda)\in\Rset^2$ be such that $p>1$ and $\lambda\neq0$. The \emph{relative Fisher divergence} is defined as
\begin{equation}\label{eq:relativeFD}
F_{p,\lambda}[f||h]:=\int_{\Omega}f^{1+p(\lambda-1)}h^{-\lambda p}\left|\frac d{dx}\left[\log\frac{f}{h}\right]\right|^p\,dx, \quad
\phi_{p,\lambda}[f||h]:=F_{p,\lambda}[f||h]^\frac{1}{p\lambda}
\end{equation}
\end{definition}

\medskip

The relative Fisher divergence is obtained as the $(p,\lambda)$-Fisher information applied to the relative differential-escort transformed density, as the following result shows.
\begin{lemma}\label{lem:relativeFD}
In the same conditions and notation as in Definition \ref{def:relativeFD}, we have
\begin{equation}\label{eq:relFI}
	\phi_{p,\lambda}[f_\alpha^{[h]}]=|\alpha|^{1/\lambda}\phi_{p,\lambda\alpha}[f||h]^{\alpha}.
	\end{equation}
\end{lemma}
\begin{proof}
Recalling the definitions \eqref{eq:def_FI} and \eqref{eq:transf} and employing Lemma \ref{lem:deriv}, we can write
\begin{equation*}
\begin{split}
\phi_{p,\lambda}^{p\lambda}[f_\alpha^{[h]}]&=\int_{\widetilde{\Omega}}\big(f_\alpha^{[h]}\big)^{1+p(\lambda-2)}(y)\big|\big(f_\alpha^{[h]}\big)'(y)\big|^p\,dy\\
&=\int_{\Omega}\left(\frac{f(x)}{h(x)}\right)^{p\alpha(\lambda-2)}\left|\frac{\alpha f^{2\alpha-1}(x)}{h^{2\alpha}(x)}\left(\frac{f'(x)}{f(x)}-\frac{h'(x)}{h(x)}\right)\right|^pf(x)\,dx\\
&=|\alpha|^p\int_{\Omega}f^{1+\alpha p(\lambda-2)+p(2\alpha-1)}(x)h^{\alpha p(2-\lambda)-2\alpha p}(x)\left|\frac{f'(x)}{f(x)}-\frac{h'(x)}{h(x)}\right|^p\,dx,\\
&=|\alpha|^p\phi_{p,\lambda\alpha}^{p\lambda\alpha}[f||h]
\end{split}
\end{equation*}
which leads to the right hand side of \eqref{eq:relFI}.
\end{proof}
Let us also remark the following symmetry relation
$$
F_{p,\lambda}[f||h]=F_{p,\widetilde\lambda}[h||f],\qquad \widetilde\lambda=1-\lambda-\frac 1p,
$$
which is easily derived from the definition.

\medskip

The second informational quantity that we introduce is called the \emph{relative cumulative moment} and extends the cumulative moments derived (see \cite{Puertas2025}) via the standard differential-escort transformations.
\begin{definition}\label{def:relativeCM}
Let $f$ and $h$ be two probability density functions satisfying \eqref{cond:support} and let $p>0$, $\alpha\in\Rset$. The \emph{relative cumulative moment} is defined as the expected value
$$
\mu_{p,\alpha}[f||h]:=\mu_p[f_\alpha^{[h]}]=\left\langle|y|^p\right\rangle=\int_{\widetilde{\Omega}}f_\alpha^{[h]}(y)|y|^p\,dy.
$$
A simple calculation employing Definition \ref{def:transf} gives the explicit expression
\begin{equation}\label{eq:relativeCM}
\mu_{p,\alpha}[f||h]=\int_{\Omega}\left|\int_{x_i}^xf(s)^{1-\alpha}h(s)^{\alpha}\,ds\right|^pf(x)\,dx.
\end{equation}
We also introduce the quantity
$$
\sigma_{p,\alpha}[f||h]:=\mu_{p,\alpha}[f||h]^\frac1{p\alpha}
$$
\end{definition}
Observe that the right hand side in \eqref{eq:relativeCM} strongly reminds of the cumulative moments defined in \cite{Puertas2025}, motivating the name. Note that, the application of such a cumulative moment $\mu_{p,\gamma}$ to a relative differential-escort transformation of $\alpha$-order gives again a relative cumulative moment $\mu_{p,\alpha\gamma}[f||h],$ in view of Eq.~\eqref{eq:comp}. It is worth mentioning that
$$
\mu_{p,\alpha}[f||f]=\langle|F(x)|^p\rangle,\qquad F(x):=\int_{x_i}^xf(t)dt,
$$
while $F_{p,\lambda}[f||f]=0$. 

\begin{remark}
In contrast to the relative Fisher information, the relative cumulative moment  $\mu_{p,\alpha}[f||h]$ is not minimized when $f=h.$ Simple counterexamples can be derived by taking $f=1$ and $h=(\eta+1)x^\eta$ in the support $\Omega=[0,1]$.
\end{remark}

A remarkable property of the latter quantities is the invariance under scaling transformations of both probability densities $f$ and $h$.
\begin{lemma}
Given $f,h:\Omega\subseteq\Rset\longrightarrow \Rset^+$ two probability densities and $r>0$ a real number, we have:
\begin{equation}
\phi_{p,\lambda}[f_{(r)}||g_{(r)}]=\phi_{p,\lambda}[f||g],\qquad \sigma_{p,\lambda}[f_{(r)}||g_{(r)}]=\sigma_{p,\lambda}[f||g]
\end{equation}
\end{lemma}
\begin{proof}
We start by observing that the support $\Omega_{(r)}$ of the scaled density $f_{(r)}=\mathcal S_{r}[f]$ is given by $\Omega_{(r)}=(x_i/r,x_f/r)$. Then, we find by direct computation that
\begin{eqnarray*}
\mu_{p,\gamma}[f_{(r)}||g_{(r)}]&=&\int_{\Omega_{(r)}}\left|\int_{x_i/r}^xf_{(r)}(s)^{1-\alpha}h_{(r)}(s)^{\alpha}\,ds\right|^pf_{(r)}(x)\,dx
\\
&=&\int_{\Omega_{(r)}}\left|\int_{x_i/r}^{x}rf(rs)^{1-\alpha}h(rs)^{\alpha}\,ds\right|^prf(rx)\,dx
\\
&=&\int_{\Omega_{(r)}}\left|\int_{x_i}^{rx}f(t)^{1-\alpha}h(t)^{\alpha}\,dt\right|^prf(rx)\,dx
\\
&=&\int_{\Omega}\left|\int_{x_i}^{\overline x}f(t)^{1-\alpha}h(t)^{\alpha}\,dt\right|^pf(\overline x)\,d\overline x
=\mu_{p,\gamma}[f||g].
\end{eqnarray*}
The scaling invariance of the relative Fisher divergence follows in a similar way.
\begin{eqnarray*}
	F_{p,\lambda}[f_{(r)}||g_{(r)}]&=&\int_{\Omega_{(r)}} f_{(r)}(x)^{1-p}\left(\frac{f_{(r)}(x)}{h_{(r)}(x)}\right)^{p\lambda}\left|\frac d{dx}\log\frac{f_{(r)}(x)}{h_{(r)}(x)}\right|^p\,dx
	\\
	&=&\int_{\Omega_{(r)}} r^{1-p}f(rx)^{1-p}\left(\frac{rf(rx)}{rh(rx)}\right)^{p\lambda}\left|\left(\log\frac{rf(rx)}{rh(rx)}\right)'\right|^p\,dx
	\\
	&=&\int_{\Omega} r^{-p}f(x)\left(\frac{f(x)}{h(x)}\right)^{p\lambda}\left|r\left(\log\frac{f(x)}{h(x)}\right)'\right|^p\,dx=F_{p,\lambda}[f||g],
\end{eqnarray*}
completing the proof.
\end{proof}
Let us stress at this point that several Fisher divergence measures have been introduced in the past. We quote first the relative Fisher information considered in \cite{Antolin2014,Yamano2021}	
\begin{equation}\label{eq:relFish}
		F_{rel}[f||g]:=\int_\Rset f(x)\left|\frac d{dx}\log\left(\frac {f(x)}{g(x)}\right)\right|^2\,dx
\end{equation}		
Other relative measures similar to the previous one have been studied and applied in the literature~\cite{Antolin2009,Martin2013,Zozor2015,Toranzo2017}. Probably the most similar one to the relative Fisher divergence that we introduce in this work, as far as we know, was given in 1978 by Hammad \cite{Hammad1978}
\begin{equation}\label{eq:relFish_Hammad}
F_{\alpha}^{H}[f||g]:=\int_\Rset f^\alpha (x)g^{1-\alpha}(x)\left|\frac d{dx}\log\left(\frac {f(x)}{g(x)}\right)\right|^2\,dx.
\end{equation}
However, these two relative Fisher measures have the rather serious drawback of not being invariant to scaling transformations. In fact, straightforward computations lead to
$$
F_{rel}[f_{(r)}||g_{(r)}]=r^2F_{rel}[f||g],\qquad F_{rel}^{H}[f_{(r)}||g_{(r)}]=r^2F^H_{rel}[f||g].
$$
This fact prevented establishing inequalities involving the two Fisher measures given in \eqref{eq:relFish} and \eqref{eq:relFish_Hammad} with the Rényi and Kullback-Leibler divergences, and indeed, to the best of our knowledge such inequalities are missing from literature. We thus consider that our definition of the relative Fisher divergence is a significant improvement with respect to these previous ones, since we are able to establish sharp inequalities connecting the relative Fisher divergence to the Rényi and Kullback-Leibler divergences, as we show in the next section.

\subsection{Informational inequalities}

We are now ready to state and prove the two informational inequalities obtained by the application of the relative differential-escort transformation.
\begin{theorem}\label{th:ineq}
Let $f$ and $h$ be two probability density functions satisfying \eqref{cond:support}, $p\in\Rset$, $\lambda\in\Rset\setminus\{1\}.$ Then, the following two inequalities hold true.

$\bullet$ \textbf{moment-entropy-like inequality} involving the relative cumulative moments and the R\'enyi divergences:  if $p^*\geqslant0,\,\alpha>0$ and $\lambda>\frac1{1+p^*}$ are satisfied, then
\begin{equation}\label{ineq:E-M-like}
e^{\, D_{\xi(\lambda,\alpha)}[f||h]}\sigma_{p^*,\alpha}[f||h]\geq (K^{(0)}_{p,\lambda})^{\frac1\alpha}.
\end{equation}

$\bullet$ \textbf{Stam-like inequality} involving the R\'enyi divergences and the relative Fisher divergence: if $p\geqslant 1,\,\alpha>0$ and the condition~\eqref{eq:sign_cond1} are satisfied, then
\begin{equation}\label{ineq:Stam-like}
	\left[e^{-D_{\xi(\lambda,\alpha)}[f||h]}\phi_{p,\beta\alpha}[f||h]\right]^{1+\beta-\lambda}
	\geq\alpha^{\frac{\lambda-\beta-1}{\alpha\beta}}\left(K^{(1)}_{p,\beta,\lambda}\right)^\frac1\alpha,
	\end{equation}
provided that $\mathfrak R^{[h]}_\alpha[f]$ is absolutely continuous and  $\alpha\neq0$.

The previous inequalities are sharp whenever $N_\frac{1}{\alpha^*}[g_{p,\lambda}]<\infty$ for \eqref{ineq:E-M-like}, respectively $N_\frac{1}{\alpha^*}[g_{p,\beta,\lambda}]<\infty$ for \eqref{ineq:Stam-like}. More precisely, the minimizers to the moment-entropy-like inequality~\eqref{ineq:E-M-like} are given by
$$
f_{\rm min}(x):=\mathfrak{R}^{-1,[h]}_{\alpha}[g_{p,\lambda}](x), \quad x\in\Omega,
$$
while the minimizers to the Stam-like inequality~\eqref{ineq:Stam-like} are given by
$$
f_{\rm min}(x):=\mathfrak{R}^{-1,[h]}_{\alpha}[g_{p,\beta,\lambda}](x), \quad x\in\Omega,
$$
according to Definition \ref{def:inverse}.
\end{theorem}
\begin{remark}
Observe that for $\beta=\lambda$ the minimizers to the inequality~\eqref{ineq:Stam-like} reduce to the same minimizers as to the inequality~\eqref{ineq:E-M-like} and these minimizers can be expressed in terms of the generalized trigonometric and hyperbolic functions and other special functions, as we shall show in the next section. In addition, still when $\beta=\lambda$, by multiplying both inequalities one obtains a Cramér-Rao-like inequality in the relative framework.
\end{remark}
\begin{remark}
Note that the condition $\mathfrak R^{[h]}_\alpha[f]$ absolutely continuous is achieved for example if $\frac fh$ is a $C^1$ function with compact support, or if $\frac fh$ is a $C^1$ function with unbounded support but whose derivative $(f/h)'$ is bounded. Also observe that the case $f=h$ is excluded from the previous condition since in this case $\mathfrak R^{[h]}_\alpha[f]$ is the uniform probability density on $[0,1],$ which is not absolutely continuous in $\Rset$.
\end{remark}
\begin{proof}
We first prove the entropy-moment-like inequality. To this end, we start from the classical entropy-moment inequality Eq. \eqref{ineq:bip_E-M} to get
\begin{equation}\label{eq:interm2}
\frac{\mu_{p^*}^{\frac{1}{p^*}}[f_\alpha^{[h]}]}{N_{\lambda}[f_\alpha^{[h]}]}\geq K^{(0)}_{p,\lambda}.
\end{equation}
We then infer from Definition \ref{def:relativeCM}, Lemma \ref{lem:Renyi} and \eqref{eq:interm2} that
$$
\frac{\mu_{p^*,\alpha}[f||h]^{\frac{1}{p^*}}}{K_{\xi(\lambda,\alpha)}[f||h]^{\frac{1}{1-\lambda}}}\geq K^{(0)}_{p,\lambda}
$$
and \eqref{ineq:E-M-like} follows by easy manipulations. Since the minimizers to the classical inequalities \eqref{ineq:bip_E-M} are the stretched Gaussians $g_{p,\lambda}$, we infer that the minimizers to the new inequality \eqref{ineq:E-M-like} are given by the inverse relative differential-escort transformation applied to $g_{p,\lambda}$.

\medskip

In order to prove the Stam-like inequality \eqref{ineq:Stam-like}, we start from the triparametric Stam inequality \eqref{ineq:trip_Stam_extended}. Since $\mathfrak R^{[h]}_\alpha [f]$ is absolutely continuous, we can thus apply the Stam inequality \eqref{ineq:trip_Stam_extended} to the transformed density $f_\alpha^{[h]}$ to get
$$
\left[N_{\lambda}[f_\alpha^{[h]}]\phi_{p,\beta}[f_\alpha^{[h]}]\right]^{1+\beta-\lambda}\geq K^{(1)}_{p,\beta,\lambda},
$$
which, in view of Lemmas \ref{lem:Renyi} and \ref{lem:relativeFD}, leads to
$$
|\alpha|^{\frac{1+\beta-\lambda}{\beta}}\left[K_{\xi(\lambda,\alpha)}^{\frac{1}{1-\lambda}}[f||h]\phi_{p,\beta,\alpha}^{\frac{1}{p\beta}}[f||h]\right]^{1+\beta-\lambda}
\geq K^{(1)}_{p,\beta,\lambda},
$$
which is equivalent to \eqref{ineq:Stam-like}. Since the minimizers to the triparametric Stam inequality \eqref{ineq:trip_Stam_extended} are the densities $g_{p,\beta,\lambda}$, we infer that the minimizers to the new inequality~\eqref{ineq:Stam-like} are given by the inverse relative differential-escort transformation applied to $g_{p,\beta,\lambda}$, as claimed.
\end{proof}
Let us remark that the inequality~\eqref{ineq:E-M-like} requires less restrictive conditions than the inequality~\eqref{ineq:Stam-like}. Indeed, in the former no assumptions of absolute continuity are needed and thus the class of densities $f,\,h$ for which it is fulfilled is more general. In particular, it also applies for $f=h$ recalling that $\mu_{p,\alpha}[f||f]=\langle F(x)^p\rangle\neq0$.

\bigskip

\noindent \textbf{Inequalities involving the Shannon entropy}. Some surprisingly simple and nice inequalities (and, up to the best of our knowledge, new) are obtained from the previous ones in the limiting case $\lambda=1$ and involve classical quantities such as the Shannon entropy and the Kullback-Leibler divergence. Although they are consequences of the inequalities established in Theorem \ref{th:ineq}, we believe that they are of sufficient interest by themselves in order to state them as a separate theorem.
\begin{theorem}\label{th:Shannon}
In the same conditions as in the statement of Theorem \ref{th:ineq}, the following two inequalities hold true:
\begin{equation}\label{ineq:E-M-like-Shannon}
\sigma_{p^*,\alpha}[f||h]e^{D[f||h]}\geq (K^{(0)}_{p,1})^{1/\alpha},
\end{equation}
and
\begin{equation}\label{ineq:Stam-like-Shannon}
e^{-D[f||h]}\phi_{p,\alpha}[f||h]\geq\left(\frac{K^{(1)}_{p,1}}{\alpha}\right)^{1/\alpha}.
\end{equation}
\end{theorem}
\begin{proof}
The inequalities \eqref{ineq:E-M-like-Shannon} and \eqref{ineq:Stam-like-Shannon} follow in a straightforward way by letting $\lambda\to1$ in the inequalities \eqref{ineq:E-M-like} and \eqref{ineq:Stam-like}, or equivalently employing the identity~\eqref{eq:Shannon}.
\end{proof}

\subsection{The minimizing densities}\label{sec:minimizers}
	
	In this section we compute the common minimizer of the two inequalities \eqref{ineq:E-M-like} and \eqref{ineq:Stam-like} when $\beta=\lambda$, assuming that 
    $$
    N_{\frac{1}{\alpha^*}}[g_{p,\lambda}]=N_{\frac{\alpha-1}{\alpha}}[g_{p,\lambda}]<\infty.
    $$ 
    Indeed, from the definition of the inverse of the relative differential-escort transformation, we deduce that
	\begin{equation}\label{eq:minim-G}
	f_{\rm min}(x)=\mathfrak{R}^{-1,[h]}_{\alpha}(g_{p,\lambda})(x)=((g_{p,\lambda})_{(r)}(y(x)))^{\frac{1}{\alpha}}h(x),
	\end{equation}
	with the change of variable
	$$
	y'(x)=((g_{p,\lambda})_{(r)}(y(x)))^{-\frac{1}{\alpha^*}}h(x),
	$$
	where $r=N_{\frac{1}{\alpha^*}}[g_{p,\lambda}]<\infty$ and $\frac{1}{\alpha^*} = \frac{\alpha-1}{\alpha}$.
	
Fix $\alpha\neq1$ and assume first that $\lambda>1$. Separating variables and integrating the differential equation implicitly defining $y'(x)$, we find
	\begin{equation*}
	\begin{split}
{\rm H}(x):&=\int_0^x h(s)\,ds=\int_0^y(rg_{p,\lambda}(rt))^{\frac{1}{\alpha^*}}\,dt\\
	&=(ra_{p,\lambda})^{\frac{1}{\alpha^*}}\int_0^{y}\bigg(1+(1-\lambda)|rt|^{p^*}\bigg)^{\frac{1}{\alpha^*(\lambda-1)}}\,dt\\ &=\frac{(ra_{p,\lambda})^{\frac{1}{\alpha^*}}}{|1-\lambda|^{\frac{1}{p^*}}r} \int_0^{|1-\lambda|^{\frac{1}{p^*}}ry}\bigg(1-\widetilde{t}^{p^*}\bigg)^{\frac{1}{\alpha^*(\lambda-1)}}\,d\widetilde{t}\\
	&=\frac{(ra_{p,\lambda})^{\frac{1}{\alpha^*}}}{|1-\lambda|^{\frac{1}{p^*}}r} \arcsin_{\alpha^*(1-\lambda),p^*}\left(|1-\lambda|^{\frac{1}{p^*}}ry\right).
	\end{split}
	\end{equation*}
	The latter expression allows us to calculate $y(x)$; more precisely, by inverting the generalized arcsine function, we obtain
	\begin{equation}\label{eq:interm4}
	|1-\lambda|^{\frac{1}{p^*}}ry(x)=\sin_{\alpha^*(1-\lambda),p^*}(C(p,\lambda,\alpha){\rm H}(x)),
	\quad C(p,\lambda,\alpha):=|1-\lambda|^{\frac{1}{p^*}}r(ra_{p,\lambda})^{-\frac{1}{\alpha^*}}.
	\end{equation}
	We next insert \eqref{eq:interm4} into \eqref{eq:minim-G}. The next simplifications are rather surprising and give us a closed expression of the minimizer $f_{\rm min}$ in terms of generalized trigonometric functions. Indeed, taking into account that $\lambda>1$, we infer from \eqref{eq:interm4} that
	\begin{equation*}
	\begin{split}
	((g_{p,\lambda})_{(r)}(y(x)))^{\frac{1}{\alpha}}&=(ra_{p,\lambda})^{\frac{1}{\alpha}}\left(1+(1-\lambda)|ry(x)|^{p^*}\right)^{\frac{1}{\alpha(\lambda-1)}}\\
	&=(ra_{p,\lambda})^{\frac{1}{\alpha}}\left(1-\sin_{\alpha^*(1-\lambda),p^*}^{p^*}(C(p,\lambda,\alpha){\rm H}(x))\right)^{\frac{1}{\alpha(\lambda-1)}}\\
	&=(ra_{p,\lambda})^{\frac{1}{\alpha}}\left(\cos_{\alpha^*(1-\lambda),p^*}^{\alpha^*(1-\lambda)}(C(p,\lambda,\alpha){\rm H}(x))\right)^{\frac{1}{\alpha(\lambda-1)}}\\
	&=(ra_{p,\lambda})^{\frac{1}{\alpha}}\left(\cos_{\alpha^*(1-\lambda),p^*}(C(p,\lambda,\alpha){\rm H}(x))\right)^{\frac{1}{1-\alpha}},
	\end{split}
	\end{equation*}
	where we have employed the generalized trigonometric identity \eqref{eq:pyth}. In the opposite case $\lambda<1$, one can go identically along the previous calculations, only replacing the generalized trigonometric functions with their hyperbolic counterparts in view of \eqref{arcsinh_def} and \eqref{eq:pythh}. We can thus express the minimizer of the inequalities \eqref{ineq:E-M-like-G} and \eqref{ineq:Stam-like-G} as follows:
	\begin{equation}\label{eq:min_gen}
	f_{\rm min}(x)=(ra_{p,\lambda})^{\frac{1}{\alpha}}\left\{\begin{array}{ll}
	\left(\cos_{\alpha^*(1-\lambda),p^*}^{\frac{1}{1-\alpha}}(C(p,\lambda,\alpha){\rm H}(x))\right)h(x), & {\rm if} \ \lambda>1, \\[2mm]
	\left(\cosh_{\alpha^*(1-\lambda),p^*}^{\frac{1}{1-\alpha}}(C(p,\lambda,\alpha){\rm H}(x))\right)h(x), & {\rm if} \ \lambda<1,
	\end{array}\right.
	\end{equation}
	where
	$$ C(p,\lambda,\alpha)=|1-\lambda|^{\frac{1}{p^*}}r\left(N_{\frac{1}{\alpha^*}}[g_{p,\lambda}]a_{p,\lambda}\right)^{-\frac{1}{\alpha^*}} = |1-\lambda|^{\frac{1}{p^*}}\left(N_{\frac{1}{\alpha^*}}[g_{p,\lambda}]\right)^{\frac{1}{\alpha}} \left(a_{p,\lambda}\right)^{-\frac{1}{\alpha^*}}.
	$$
In the case $\lambda=1$, the minimizers are given by
\[
g_{p,1} (x) = a_{p,1}e^{-|x|^{p^*}}.
\]
Following the same steps as for the cases $\lambda>1$ and $\lambda <1$, we first separate variables and integrate the differential equation associated to the change of variable,
	\begin{equation*}
	\begin{split}
		{\rm H}(x):&=\int_0^x h(s)\,ds=\int_0^y(rg_{p,\lambda}(rt))^{\frac{1}{\alpha^*}}\,dt\\
		&= (ra_{p,1})^{\frac{1}{\alpha^*}} \int_0^y  e^{-\frac{|rt|^{p^*}}{\alpha^*}} \,dt\\
		&= (ra_{p,1})^{\frac{1}{\alpha^*}} \frac{(\alpha^*)^{\frac{1}{p^*} } }{r\, p^* }\int_0^{\frac{r^{p^*} }{\alpha^* }|y|^{p^*} }  \widetilde{t}^{\frac{1 }{p^* }-1 } e^{-\widetilde{t}} \,d\widetilde{t}\\
		&= (ra_{p,1})^{\frac{1}{\alpha^*}} \frac{(\alpha^*)^{\frac{1}{p^*} } }{r\, p^* }\gamma\left(\frac{1}{p^*}, \frac{r^{p^*} }{\alpha^* }|y|^{p^*}  \right),
	\end{split}
\end{equation*}

\noindent where $\gamma(s,x)$ denotes the incomplete Gamma function $\gamma(s,x)=\int_0^x t^{s-1}e^{-t}\,dt$. Then, inverting the incomplete Gamma function, we have
\begin{equation}\label{eq:interm7}
r^{p^*} |y|^{p^*} =\alpha^*\gamma^{-1}\left(\frac{1}{p^*}, C(p,1,\alpha){\rm H}(x) \right),
	\quad C(p,1,\alpha):= (ra_{p,1})^{-\frac{1}{\alpha^*}} \frac{r\, p^*}{(\alpha^*)^{\frac{1}{p^*} } }.
\end{equation}
Inserting~\eqref{eq:interm7} into~\eqref{eq:minim-G} gives
	\begin{equation*}
	\begin{split}
		((g_{p,1})_{(r)}(y(x)))^{\frac{1}{\alpha}}&=(ra_{p,1})^{\frac{1}{\alpha}} \left(\exp\left[-|ry(x)|^{p^*}\right]\right)^{\frac{1}{\alpha}} \\
		&=(ra_{p,1})^{\frac{1}{\alpha}} \left(\exp\left[-\alpha^* \gamma^{-1}\left(\frac{1}{p^*}, C(p,1,\alpha){\rm H}(x) \right)\right]\right)^{\frac{1}{\alpha}} \\
		&=(ra_{p,1})^{\frac{1}{\alpha}}\exp\left[-\frac{\alpha^*}{\alpha}\gamma^{-1}\left(\frac{1}{p^*}, C(p,1,\alpha){\rm H}(x)\right)\right] \\
	\end{split}
\end{equation*}
Thus the minimizer of the inequalities \eqref{ineq:E-M-like-G} and \eqref{ineq:Stam-like-G} for $\lambda =1$ is given by
\begin{equation}\label{eq:min_gen1}
	f_{\rm min}(x)=(ra_{p,1})^{\frac{1}{\alpha}}\exp\left[-\frac{\alpha^*}{\alpha}\gamma^{-1}\left(\frac{1}{p^*}, C(p,1,\alpha){\rm H}(x) \right)\right] h(x)
\end{equation}
where
$$
C(p,1,\alpha)=  \left(N_{\frac{1}{\alpha^*}}[g_{p,\lambda}]\right)^{\frac{1}{\alpha}} \left(a_{p,1}\right)^{-\frac{1}{\alpha^*}} \frac{ p^*}{(\alpha^*)^{\frac{1}{p^*} } }.
$$
Finally, when $\alpha=1,$ the change of variable reduces to $y(x)={\rm H}(x),$ and then, in the case when $g_{p,\lambda}$ is compactly supported (which corresponds to $\lambda>1$), the minimizing density writes
\begin{equation}\label{eq:min_gen_a1}
    f_{min}(x)=rg_{p,\lambda}(r {\rm H}(x))h(x),
\end{equation}
where $r$ is the length of the support of $g_{p,\lambda}$.

\subsection{An example with the Gaussian distribution}\label{sec:Gauss}

In this section we give an example of the previous theory, by picking the Gaussian density function (for simplicity, with expected value zero and variance $1/2$), namely
$$
G(x)=\frac{1}{\sqrt{\pi}}e^{-x^2}, \quad x\in\Rset
$$
as reference function. Letting thus $h=G$, we find by direct calculation from \eqref{eq:relativeFD} that
$$
F_{p,\lambda}[f||G]=\pi^{\frac{p\lambda}{2}}\int_{\Rset}f^{1+p(\lambda-1)}(x)e^{p\lambda x^2}\left|\frac{f'}{f}+2x\right|^p\,dx
$$
while \eqref{eq:relativeCM} gives
$$
\mu_{p,\alpha}[f||G]=\pi^{-\frac{\alpha p}{2}}\int_{\Rset}\left|\int_0^xf^{1-\alpha}(t)e^{-\alpha t^2}\,dt\right|^pf(x)\,dx.
$$
Finally, the R\'enyi divergence (for $\xi\neq1$) is given by
$$
D_{\xi}[f||g]=\frac{1}{\xi-1}\log\left(\int_{\Rset}f^{\xi}(x)\frac{e^{-(1-\xi)x^2}}{\pi^{(1-\xi)/2}}\,dx\right)
=\frac{1}{\xi-1}\log\left(\int_{\Rset}f^{\xi}(x)e^{(\xi-1)x^2}\,dx\right)+\frac{1}{2}\log\pi.
$$
Thus, in this particular case the inequalities \eqref{ineq:E-M-like} and \eqref{ineq:Stam-like} become
\begin{equation}\label{ineq:E-M-like-G}
\left(\int_{\Rset}f^{\xi(\lambda,\alpha)}(x)e^{\alpha(\lambda-1)x^2}\,dx\right)^{\frac{1}{\alpha(\lambda-1)}}
\left(\int_{\Rset}\left|\int_0^xf^{1-\alpha}(t)e^{-\alpha t^2}\,dt\right|^{p^*}f(x)\,dx\right)^{\frac{1}{p^*\alpha}}\geq K,
\end{equation}
respectively, letting for simplicity $\beta=\lambda$,
\begin{equation}\label{ineq:Stam-like-G}
\left(\int_{\Rset}f^{\xi(\lambda,\alpha)}(x)e^{\alpha(\lambda-1)x^2}\,dx\right)^{\frac{1}{\alpha(1-\lambda)}}
\left(\int_{\Rset}f^{1+p(\lambda\alpha-1)}(x)e^{p\lambda\alpha x^2}\left|\frac{f'}{f}+2x\right|^p\,dx\right)\geq K,
\end{equation}
where $K$ denotes a generic positive constant which can be made explicit in view of the general constants of the inequalities \eqref{ineq:E-M-like}, \eqref{ineq:Stam-like} and the factors of powers of $\pi$ appearing in the functionals for $h=G$.

\medskip

If $\alpha\neq1$ and $\lambda\neq1$, by particularizing \eqref{eq:min_gen} to $h(x)=G(x)$, we can thus express the minimizer of the inequalities \eqref{ineq:E-M-like-G} and \eqref{ineq:Stam-like-G} as follows:
\begin{equation*}
f_{\rm min}(x)=\frac{(ra_{p,\lambda})^{\frac{1}{\alpha}}}{\sqrt{\pi}}\left\{\begin{array}{ll}
\left(\cos_{\alpha^*(1-\lambda),p^*}^{\frac{1}{1-\alpha}} \left(\frac12C(p,\lambda,\alpha){\rm erf}(x)\right)\right)e^{-x^2}, & {\rm if} \ \lambda>1, \\[2mm]
\left(\cosh_{\alpha^*(1-\lambda),p^*}^{\frac{1}{1-\alpha}}\left(\frac12C(p,\lambda,\alpha){\rm erf}(x)\right)\right)e^{-x^2}, & {\rm if} \ \lambda<1,
\end{array}\right.
\end{equation*}
while for $\lambda=1$ and $\alpha\neq1$, the minimizer is deduced by particularizing \eqref{eq:min_gen1} to $h(x)=G(x)$, namely
$$
f_{\rm min}(x)=\frac{(ra_{p,1})^{\frac1\alpha}}{\sqrt{\pi}}\exp\left[-\frac{\alpha^*}{\alpha}\gamma^{-1}\left(\frac{1}{p^*}, \frac{1}{2}C(p,1,\alpha){\rm erf}(x) \right)\right]e^{-x^2}.
$$
Finally, for $\alpha=1$, the minimizer is deduced by letting $h(x)=G(x)$ in \eqref{eq:min_gen_a1}, that is,
$$
f_{\rm min}(x)=\frac{r}{\sqrt{\pi}}g_{p,\lambda}\left(\frac{r}{2}{\rm erf}(x)\right)e^{-x^2}.
$$

\subsection{A monotonicity property for $p=2$ and $\lambda=1$}

We extend below a monotonicity property of the Fisher-Shannon complexity measure with respect to the convolution operator with a Gaussian function, to a new relative Fisher-Shannon complexity measure constructed starting from our relative differential-escort transformation. Given two probability density functions $f$ and $h$ satisfying \eqref{cond:support}, we introduce thus the relative Fisher-Shannon complexity measure
\begin{equation}\label{def:CFS}
C_{FS,\alpha}^{[h]}[f]=C_{FS,\alpha}[f||h]:=C_{FS}(\mathfrak{R}_{\alpha}^{[h]}[f])=e^{-D[f||h]}\phi_{2,\alpha}[f||h],
\end{equation}
where $C_{FS}$ is the standard Fisher-Shannon complexity measure. Denoting by
$$
G_{\tau}(x):=\frac{1}{\sqrt{2\pi\tau}}e^{-\frac{x^2}{2\tau}}, \quad \mathfrak{C}_{\tau}^{G}[f]:=G_{\tau}* f
$$
the Gaussian density function with variance $\tau$ and the convolution operator with this Gaussian density, respectively, we have the following monotonicity property.
\begin{proposition}\label{prop:monot}
In the previous conditions and notation, we have
$$
C_{FS,\alpha}^{[h]}[\mathfrak{R}_{\alpha}^{-1,[h]}\mathfrak{C}_{\tau}^{G}\mathfrak{R}_{\alpha}^{[h]}[f]]\leq C_{FS,\alpha}^{[h]}[f].
$$
Moreover, when $f=\mathfrak{R}_{\alpha}^{-1,[h]}[G_{\nu}],$ for an arbitrary variance $\nu,$ then the equality holds for any value of $\alpha.$
\end{proposition}
\begin{proof}
The cornerstone of the proof is the following monotonicity property for the standard Fisher-Shannon complexity, established in \cite{Rudnicki2016}:
\begin{equation}\label{eq:interm5}
C_{FS}[\mathfrak{C}_{\tau}^{G}[f]]\leq C_{FS}[f].
\end{equation}
We also deduce from \eqref{def:CFS} that
\begin{equation}\label{eq:interm6}
C_{FS}[f]=C_{FS,\alpha}^{[h]}[\mathfrak{R}_{\alpha}^{-1,[h]}[f]].
\end{equation}
Putting together Eqs. \eqref{eq:interm5} and \eqref{eq:interm6}, we can write
\begin{equation*}
\begin{split}
C_{FS,\alpha}^{[h]}[\mathfrak{R}_{\alpha}^{-1,[h]}\mathfrak{C}_{\tau}^{G}\mathfrak{R}_{\alpha}^{[h]}[f]]&=C_{FS}[\mathfrak{C}_{\tau}^{G}\mathfrak{R}_{\alpha}^{[h]}[f]]\\
&\leq C_{FS}[\mathfrak{R}_{\alpha}^{[h]}[f]]=C_{FS,\alpha}^{[h]}[f],
\end{split}
\end{equation*}
as claimed. If $f=\mathfrak{R}_{\alpha}^{-1,[h]}[G_{\nu}]$, then the left side of the inequality in Proposition~\ref{prop:monot} reads
$$
C_{FS,\alpha}^{[h]}[\mathfrak{R}_{\alpha}^{-1,[h]}\mathfrak{C}_{\tau}^{G}\mathfrak{R}_{\alpha}^{[h]}[\mathfrak{R}_{\alpha}^{-1,[h]}[G_{\nu}]]]
=C_{FS}[\mathfrak{C}_{\tau}^{G}[G_{\nu}]]=C_{FS}[G_{\widetilde\nu}],
$$
with $\widetilde{\nu}=\sqrt{\tau^2+\nu^2}$, while the right side reads
$$
C_{FS,\alpha}^{[h]}[\mathfrak{R}_{\alpha}^{-1,[h]}[G_{\nu}]]=C_{FS}[G_{\nu}].
$$
It is well known that $C_{FS}[G_{\nu}]=C_{FS}[G_{\widetilde\nu}]$ and their common value is the minimal value of the Fisher-Shannon complexity measure~\cite{Rudnicki2016}, concluding the proof.
\end{proof}

\section{An inverse problem: adapted sharp inequalities}\label{sec:inverse}

In this section, we apply the previous results to solve the inverse problem of finding the reference function $h^\star$ and construct inequalities whose minimizing density, which we denote by $f^\star$, is fixed \emph{a priori}. An interesting fact related to these inequalities is that the generalized trigonometric functions play again a central role, in a somehow similar manner as in the example when the reference density is the Gaussian function in Section \ref{sec:Gauss}. Let us begin by imposing that the minimizing density is
$$
f_{min}=\mathfrak{R}_\alpha^{-1,h}[g_{p,\lambda}]=:f^\star.
$$
We then have
$$
h(x)=f^\star(x)(g_{p,\lambda})_{(r)}^{-\frac1\alpha}(y(x))
$$
and thus
$$
y'(x)=(g_{p,\lambda}(y))_{(r)}^{-\frac{1}\alpha^*}h(x)=(g_{p,\lambda}(y))_{(r)}^{-1}f^\star(x).
$$
We then readily obtain by integration
\begin{equation}\label{eq:cond_star}
\int r g_{p,\lambda}(ry)\,dy=\int f^*(x)dx\equiv F^\star(x),
\end{equation}
where $F^\star$ denotes the cumulative function of $f^\star$. In the next lines we derive an explicit expression for the change of variable.

\medskip

$\bullet$ \textbf{Case $\lambda\neq1$.} Let us assume first, for simplicity, that $\lambda>1$. We start by observing that
$$
\int_0^y r g_{p,\lambda}(rt)\,dt
=\frac{a_{p,\lambda}}{|1-\lambda|^\frac 1{p^*}}\int_0^{r|1-\lambda|^\frac 1{p^*} y} (1-t^{p^*})^{\frac1{\lambda-1}}dt
=\frac{a_{p,\lambda}}{|1-\lambda|^\frac 1{p^*}}\arcsin_{1-\lambda,p^*}\left(r|1-\lambda|^\frac 1{p^*} y\right),
$$
where $\arcsin_{v,w}$ denotes the generalized arcsin function, see Section \ref{sec:GTFs}. We readily deduce that
$$
|1-\lambda|^\frac 1{p^*} ry(x)=
\sin_{1-\lambda,p^*}\left(\frac{|1-\lambda|^\frac 1{p^*}}{a_{p,\lambda}} F^\star(x)\right).
$$
We continue by computing $g_{p,\lambda}(y(x))$, taking into account the generalized trigonometric identity \eqref{eq:pyth}. To simplify the notation, we set
$$
z(x):=\frac{|1-\lambda|^\frac 1{p^*}}{a_{p,\lambda}} F^\star(x).
$$
We then have
\begin{eqnarray*}
r g_{p,\lambda}(ry(x))&=& ra_{p,\lambda}\left(1-\sin_{1-\lambda,p^*}^{p^*}(z(x))\right)^{\frac1{\lambda-1}}
=ra_{p,\lambda}\left(\cos_{1-\lambda,p^*}^{1-\lambda}(z(x))\right)^{\frac1{\lambda-1}}
\\
&=&ra_{p,\lambda}\left(\cos_{1-\lambda,p^*}\left(\frac{|1-\lambda|^\frac 1{p^*}}{a_{p,\lambda}} F^\star(x)\right)\right)^{-1}\equiv
\mathcal C_{p,\lambda}(f^\star;x)^{-1}.
\end{eqnarray*}
Keeping the previous notation, it follows that
\begin{equation}\label{eq:hstar}
h^\star(x)=f^\star(x)\mathcal C_{p,\lambda}(f^\star;x)^{\frac1\alpha}.
\end{equation}
We continue by calculating the Kullback-Leibler divergence relative to $h^\star$, employing \eqref{eq:hstar} to find
\begin{equation}\label{def:kullback-leibler_star}
D[f||h^\star]:= \int_{\Omega}  f(t)\log\left[\frac{f(t)}{f^\star(t)}\right]\,dt-\frac{1}{\alpha}\left\langle\log \mathcal C_{p,\lambda}(f^\star;x)\right\rangle
=D[f||f^\star]-\frac{1}{\alpha}\left\langle\log \mathcal C_{p,\lambda}(f^\star;x)\right\rangle.
\end{equation}
We next introduce the following notation
\begin{equation}\label{def:star}
\begin{split}
&\sigma^\star_{p,\alpha}[f^\star;f]:=\sigma_{p, \alpha}[f||h^\star]=\left[\int_{\Omega}\left|\int_{x_i}^xf(s)^{1-\alpha}(f^\star(x))^\alpha\mathcal   C_{p,\lambda}(f^\star;x)\,ds\right|^pf(x)\,dx\right]^{\frac1{p\alpha}}
\\
&\phi^\star_{p,\beta,\lambda,\alpha}[f^\star;f]:=\phi_{p,\beta}[f||h^\star]\\
&=\left[\int_{\Omega}\frac{f^{1+p(\beta-1)}}{(f^\star)^{\beta p}\mathcal C_{p,\lambda}(f^\star;x)^{\frac {p\beta}\alpha}}
\left|\frac d{dx}\left[\log\frac{f^\star(x)\mathcal C_{p,\lambda}(f^\star;x)^{\frac1\alpha}}{f}\right]\right|^p\,dx\right]^\frac1{p\beta},
\\
&K^\star_{\xi, p,\lambda,\alpha}[f^\star;f]:=e^{D_\xi[f||h^\star]}=\left[\int_{\Omega}f^\xi(x)(f^{\star}(x))^{1-\xi}\mathcal C_{p,\lambda}(f^\star;x)^{\frac{1-\xi}\alpha}dx\right]^{\frac1{\xi-1}}.
\end{split}
\end{equation}
In view of these notations, the inequality~\eqref{ineq:E-M-like} becomes
\begin{equation}\label{ineq:E-M-like_star}
K^\star_{\xi(\lambda,\alpha), p,\lambda,\alpha}[f^\star;f]\sigma_{p^*,\alpha}^\star[f^\star;f]\geq (K^{(0)}_{p,\lambda})^{\frac1\alpha},
\end{equation}
while the Stam inequality~\eqref{ineq:Stam-like} is rewritten as
\begin{equation}\label{ineq:Stam-like_star}
\frac{\phi^\star_{p,\lambda\alpha,\lambda,\alpha}[f^\star;f]}{K^\star_{\xi(\lambda,\alpha), p,\lambda,\alpha}[f^\star;f]}
\geq\alpha^{\frac{-1}{\alpha\lambda}}\left(K^{(1)}_{p,\lambda}\right)^\frac1\alpha.
\end{equation}
Note that although the latter quantities look like divergence measures, they do not fulfill one of the basic properties of a divergence, that is, to be minimized in the case $f=f^\star$ (for example $K^\star_{\xi(\lambda,\alpha), p,\lambda,\alpha}[f^\star;f]$ is minimized when $f=h^\star$). However, both inequalities~\eqref{ineq:E-M-like_star} and~\eqref{ineq:Stam-like_star} are only minimized when $f=f^\star$. At the same time, one may easily check that both left hand side of inequalities~\eqref{ineq:E-M-like_star} and~\eqref{ineq:Stam-like_star} satisfy the invariance and boundedness properties required by some authors to the so-called statistical complexity measures~\cite{Rudnicki2016} whose minimizing density is $f^\star$. In the case $\lambda<1$, all the previous lines can be performed as well, but employing the hyperbolic counterparts of the generalized trigonometric functions.

\medskip

$\bullet$ \textbf{Case $\lambda=1$.} In this paragraph we explore analogous adapted inequalities with $\lambda=1$. Let us start with the simpler case $p=2$. In this case the condition~\eqref{eq:cond_star} is written as
$$
\frac{1}{\sqrt\pi}\int_{0}^{y} e^{-(rt)^2}r\,dt=\frac{{\rm erf}(r y)}2=F^*(x),
$$
that is,
$$
ry(x)={\rm erf}^{-1}\left(2F^*(x)\right).
$$
Then,
\begin{eqnarray*}
	h^\star(x)&=&f^\star(x)\left(G_{(r)}(y(x))\right)^{-\frac 1\alpha} =f^\star(x)\left( r G(ry(x))\right)^{-\frac 1\alpha} \\ &=&f^\star(x) r^{-\frac{1}{\alpha}}\pi^{\frac{1}{2\alpha }}\exp\left[\frac{1}{\alpha}\,{\rm erf}^{-1}\left(2 F^*(x)\right)^2\right]
	\\
	&=& f^\star(x)\mathcal C_{2,1}^{\frac1\alpha}(f^\star;x)
\end{eqnarray*}
where
\begin{equation*}
C_{2,1}(f^\star;x)=r^{-1}\pi^{\frac{1}{2}}\exp\left[{\rm erf}^{-1}\left(2 F^*(x)\right)^2\right].
\end{equation*}
In the case $p\neq2$, similar calculations lead to
\begin{equation*}
h^\star(x)=C_{p,1}^{\frac{1}{\alpha}}(f^\star;x)f^\star(x), \quad C_{p,1}(f^\star;x)=(ra_{p,1})^{-1}\exp\left[\gamma^{-1}\left(\frac{1}{p^*};\frac{p^* }{a_{p,1}}F^*(x)\right)\right].
\end{equation*}
The definitions~\eqref{def:kullback-leibler_star} and~\eqref{def:star} remain in force by simply replacing $C_{p,\lambda}$ by $C_{p,1}$ in the right hand side of them. With this notation the inequalities~\eqref{ineq:E-M-like-Shannon} and~\eqref{ineq:Stam-like-Shannon} become
\begin{equation}\label{ineq:E-M-like-Shannon_star}
\sigma_{p^*,\alpha}^\star[f^\star; f]e^{D[f||f^\star]}\geq (K^{(0)}_{p,1})^{1/\alpha}e^{\frac{1}{\alpha}\left\langle\log \mathcal C_{p,1}(f^\star;x)\right\rangle},
\end{equation}
and
\begin{equation}\label{ineq:Stam-like-Shannon_star}
\frac{\phi_{p,1,1,\alpha}^\star(f^\star;f)}
{e^{D[f||f^\star]-\frac{1}{\alpha}\left\langle\log \mathcal C_{p,1}(f^\star;x)\right\rangle}}
\geq\left(\frac{K^{(1)}_{p,1}}{\alpha}\right)^{1/\alpha}.
\end{equation}

We end this section by noticing the quite interesting fact that, in the previous adapted inequalities~\eqref{ineq:Stam-like_star} and~\eqref{ineq:Stam-like-Shannon_star}, each of the two  components involved in them is minimized only by $f=h^\star$, while the inequalities are minimized only by $f=f^\star$. For comparison, let us recall that, in the classical moment-entropy and Stam inequalities, each of the factors is minimized by a different density (the $p$-th deviation by a Dirac delta, the Rényi or Shannon entropy by the uniform distribution and the Fisher information by power-like and exponential densities), while in our case the minimizer of the two factors taken separately is the same function. The fact that the minimizer of the inequalities \eqref{ineq:Stam-like_star} and \eqref{ineq:Stam-like-Shannon_star} is a different function than this common minimizer of the components is thus an even more striking property of these inequalities.

\section{Conclusions}\label{sec:conclusion}

In this work we have introduced the relative differential-escort transformation, which is a completely non-trivial generalization of the differential-escort transformation, since it is related to a reference density function $h$. It is, in fact, a large family of transformations, as by modifying the density $h$ we obtain a different outcome. With the help of these transformations, we have introduced new relative measures such as the relative Fisher divergence and the relative cumulative moment. A remarkable novelty related to these measures is that they are invariant to scaling changes. Up to the best of our knowledge, such a scale invariance is not fulfilled by any of the previous definitions of Fisher divergence measures introduced in the literature, a drawback that prevented the researchers in establishing sharp inequalities connecting these measures with other informational functionals.

Once defined these measures and proved their invariance with respect to scaling changes, we established sharp Stam-like and moment-entropy like informational inequalities relating the previously mentioned measures with well-established divergences such as the R\'enyi divergence. Another remarkable fact related to these inequalities is that, when taking the limit $\lambda\to1$ in the R\'enyi divergence, we derive very simple and elegant inequalities involving classical quantities such as the Kullback-Leibler.

A number of particular examples of interest for the theory in the application of the relative differential-escort transformation are given, involving exponential densities or the Gaussian density function as reference function. In the latter case, the generalized trigonometric and hyperbolic functions appear in a surprising way. Moreover, we established in the paper some adapted informational inequalities by fixing the future minimizer $f^\star$ and solving the inverse problem of computing the reference function $h^\star$, as well as constructing functionals related to this previously fixed density function and the calculated reference function. Once more, the generalized trigonometric and hyperbolic functions, as well as other renowned special functions (such as the incomplete Gamma function) play an important role in the mathematical expression of the reference function $h^\star$ leading to the given, fixed minimizer $f^\star$.

We believe that many applications in information theory and other fields of science could be derived in the future by applying the relative differential-escort transformation with respect to suitably chosen reference functions $h$, according to each case of interest, and by employing the sharp informational inequalities established in this paper. The rather free choice of $h$ provides a wide amplitude of possibilities of getting new information about the density $f$. We thus propose this family of transformations as a possibly strong tool for further investigation.

\subsection*{Acknowledgements}
	
All authors are partially supported by the Grants PID2020-115273GB-I00 and RED2022-134301-T funded by \text{MCIN/AEI/10.13039/501100011033}. D. P.-C. is also partially supported by the Grant PID2023-153035NB-100.

\bigskip

\noindent \textbf{Data availability} Our manuscript has no associated data.

\bigskip

\noindent \textbf{Competing interest} The authors declare that there is no competing interest.


\bibliographystyle{unsrt}
\bibliography{refs}

\begin{thebibliography}{10}

\bibitem{Fisher1922}
R.~A. Fisher.
\newblock On the mathematical foundations of theoretical statistics.
\newblock {\em Philosophical Transactions of the Royal Society of London A},
  222(594-604):309--368, 1922.

\bibitem{Shannon1949}
C.~E. Shannon and W.~Weaver.
\newblock A mathematical model of communication.
\newblock {\em Urbana, IL: University of Illinois Press}, 11:11--20, 1949.

\bibitem{Kolmogorov1956}
A.~Kolmogorov.
\newblock On the {S}hannon theory of information transmission in the case of
  continuous signals.
\newblock {\em IRE Transactions on Information Theory}, 2(4):102--108, 1956.

\bibitem{Renyi1961}
A.~R{\'e}nyi.
\newblock On measures of entropy and information.
\newblock In {\em Proceedings of the fourth Berkeley symposium on mathematical
  statistics and probability, volume 1: contributions to the theory of
  statistics}, volume~4, pages 547--562. University of California Press, 1961.

\bibitem{Kapur1969}
J.~N. Kapur.
\newblock Some properties of entropy of order $\alpha$ and type $\beta$.
\newblock In {\em Proceedings of the Indian Academy of Sciences-Section A},
  volume~69, pages 201--211. Springer, 1969.

\bibitem{Taneja1989}
I.~J. Taneja, L.~Pardo, D.~Morales, and M.~L. Men{\'e}ndez.
\newblock On generalized information and divergence measures and their
  applications: A brief review.
\newblock {\em Q{\"u}estii{\'o}: quaderns d'estad{\'\i}stica i investigaci{\'o}
  operativa}, 1989.

\bibitem{Taneja1989(libro)}
I.~J. Taneja.
\newblock On generalized information measures and their applications.
\newblock In {\em Advances in Electronics and Electron Physics}, volume~76,
  pages 327--413. Elsevier, 1989.

\bibitem{Quesada-Taneja1994}
V.~Quesada and I.~J. Taneja.
\newblock Generalized mean of order $t$ via {B}ox and {C}ox's transformation.
\newblock {\em Tamkang Journal of Mathematics}, 25(2):125--141, 1994.

\bibitem{Ilic2021}
V.~M. Ili{\'c}, J.~Korbel, S.~Gupta, and A.~M. Scarfone.
\newblock An overview of generalized entropic forms (a).
\newblock {\em Europhysics Letters}, 133(5):50005, 2021.

\bibitem{Lutwak2005}
E.~Lutwak, D.~Yang, and G.~Zhang.
\newblock Cram\'er--{R}ao and moment-entropy inequalities for \text{R}\'enyi
  entropy and generalized {F}isher information.
\newblock {\em IEEE Transactions on Information Theory}, 51(2):473--478, 2005.

\bibitem{Bercher2012a}
J.-F. Bercher.
\newblock On a {$(\beta,q)-$}generalized {F}isher information and inequalities
  involving {$q-$}{G}aussian distributions.
\newblock {\em Journal of Mathematical Physics}, 53(6), 2012.

\bibitem{Rao2004}
M.~Rao, Y.~Chen, B.~C. Vemuri, and F.~Wang.
\newblock Cumulative residual entropy: a new measure of information.
\newblock {\em IEEE Transactions on Information Theory}, 50(6):1220--1228,
  2004.

\bibitem{Tempesta2019}
M.~{\'A}. Rodr{\'\i}, {\'A}.~Romaniega, and P.~Tempesta.
\newblock A new class of entropic information measures, formal group theory and
  information geometry.
\newblock {\em Proceedings of the Royal Society A}, 475(2222):20180633, 2019.

\bibitem{Zozor2015}
S.~Zozor and J.-M. Brossier.
\newblock de {B}ruijn identities: From {S}hannon, {K}ullback--{L}eibler and
  {F}isher to generalized {$\varphi-$}entropies, {$\varphi-$}divergences and
  {$\varphi-$F}isher informations.
\newblock In {\em AIP Conference Proceedings}, volume 1641, pages 522--529,
  2015.

\bibitem{Toranzo2017}
E.~V. Toranzo, S.~Zozor, and J.-M. Brossier.
\newblock Generalization of the de {B}ruijn identity to general
  $\phi$-entropies and $\phi$-{F}isher informations.
\newblock {\em IEEE Transactions on Information Theory}, 64(10):6743--6758,
  2017.

\bibitem{Kharazmi2023}
O.~Kharazmi, H.~Jamali, and J.~E. Contreras-Reyes.
\newblock Fisher information and its extensions based on infinite mixture
  density functions.
\newblock {\em Physica A: Statistical Mechanics and Its Applications},
  624:128959, 2023.

\bibitem{Kennard1927}
E.~H. Kennard.
\newblock Zur quantenmechanik einfacher bewegungstypen.
\newblock {\em Zeitschrift für Physik}, 44:326--352, 1927.

\bibitem{Robertson1929}
H.~P. Robertson.
\newblock The uncertainty principle.
\newblock {\em Physical Review}, 34:573--574, 1929.

\bibitem{Schrodinger1930}
E.~Schrödinger.
\newblock Zum heisenbergschen unscharfeprinzip.
\newblock {\em Berliner Berichte}, pages 296--303, 1930.

\bibitem{Bialynicki1975}
I.~Bia{\l}ynicki-Birula and J.~Mycielski.
\newblock Uncertainty relations for information entropy in wave mechanics.
\newblock {\em Communications in Mathematical Physics}, 44(2):129--132, 1975.

\bibitem{Bialynicki1984}
I.~Bialynicki-Birula.
\newblock Entropic uncertainty relations.
\newblock {\em Physics Letters A}, 103(5):253--254, 1984.

\bibitem{Bialynicki2006}
I.~Bialynicki-Birula.
\newblock Formulation of the uncertainty relations in terms of the {R}{\'e}nyi
  entropies.
\newblock {\em Physical Review A—Atomic, Molecular, and Optical Physics},
  74(5):052101, 2006.

\bibitem{Romera2006}
E.~Romera, P.~S{\'a}nchez-Moreno, and J.~S. Dehesa.
\newblock Uncertainty relation for {F}isher information of d-dimensional
  single-particle systems with central potentials.
\newblock {\em Journal of mathematical physics}, 47(10), 2006.

\bibitem{Dembo1991}
A.~Dembo, T.~M. Cover, and J.~A. Thomas.
\newblock Information theoretic inequalities.
\newblock {\em IEEE Transactions on Information theory}, 37(6):1501--1518,
  1991.

\bibitem{Tsallis2022b}
C.~Tsallis.
\newblock Enthusiasm and skepticism: Two pillars of science—a nonextensive
  statistics case.
\newblock {\em Physics}, 4(2):609--632, 2022.

\bibitem{Sason2022}
I.~Sason.
\newblock Divergence measures: Mathematical foundations and applications in
  information-theoretic and statistical problems.
\newblock {\em Entropy}, 24(5), 2022.

\bibitem{Kullback1951}
S.~Kullback and R.~A. Leibler.
\newblock On information and sufficiency.
\newblock {\em The annals of mathematical statistics}, 22(1):79--86, 1951.

\bibitem{Lin2002}
J.~Lin.
\newblock Divergence measures based on the {S}hannon entropy.
\newblock {\em IEEE Transactions on Information theory}, 37(1):145--151, 2002.

\bibitem{Hammad1978}
P.~Hammad.
\newblock Mesure d’ordre $\alpha$ de l’information au sens de {F}isher.
\newblock {\em Revue de statistique appliqu{\'e}e}, 26(1):73--84, 1978.

\bibitem{Singh1997}
V.~P. Singh.
\newblock The use of entropy in hydrology and water resources.
\newblock {\em Hydrological processes}, 11(6):587--626, 1997.

\bibitem{Karmeshu2003}
Karmeshu.
\newblock {\em Entropy measures, maximum entropy principle and emerging
  applications}.
\newblock Springer Science \& Business Media, 2003.

\bibitem{Singh2013}
V.~P. Singh.
\newblock {\em Entropy theory and its application in environmental and water
  engineering}.
\newblock John Wiley \& Sons, 2013.

\bibitem{Sen2011}
K.~D. Sen.
\newblock {\em Statistical complexity: applications in electronic structure}.
\newblock Springer Science \& Business Media, 2011.

\bibitem{Sen2014}
K.~D. Sen and K.~D. Sen.
\newblock {\em Electronic structure of quantum confined atoms and molecules},
  volume~10.
\newblock Springer, 2014.

\bibitem{Harremoes2014}
T.~van Erven and P.~Harremoes.
\newblock {R}\'enyi {D}ivergence and {K}ullback-{L}eibler {D}ivergence.
\newblock {\em IEEE Transactions on Information Theory}, 60(7):3797--3820,
  2014.

\bibitem{Antolin2009}
J.~Antol{\'\i}n, J.~C. Angulo, and S.~L{\'o}pez-Rosa.
\newblock Fisher and {J}ensen--{S}hannon divergences: Quantitative comparisons
  among distributions. application to position and momentum atomic densities.
\newblock {\em The Journal of chemical physics}, 130(7), 2009.

\bibitem{Martin2015}
A.~L. Mart{\'\i}n, S.~L{\'o}pez-Rosa, J.~C. Angulo, and J.~Antol{\'\i}n.
\newblock Jensen--{S}hannon and {K}ullback--{L}eibler divergences as
  quantifiers of relativistic effects in neutral atoms.
\newblock {\em Chemical Physics Letters}, 635:75--79, 2015.

\bibitem{Toscani2016(a)}
G.~Toscani.
\newblock Entropy inequalities for stable densities and strengthened central
  limit theorems.
\newblock {\em Journal of Statistical Physics}, 165(2):371--389, 2016.

\bibitem{Toscani2017}
G.~Toscani.
\newblock Score functions, generalized relative {F}isher information and
  applications.
\newblock {\em Ricerche di Matematica}, 66(1):15--26, 2017.

\bibitem{Puertas2025}
D.~Puertas-Centeno and S.~Zozor.
\newblock Some informational inequalities involving generalized trigonometric
  functions and a new class of generalized moments.
\newblock {\em Journal of Physics A: Mathematical and Theoretical},
  58(16):165002, 2025.

\bibitem{Rudnicki2016}
{\L}.~Rudnicki, E.~V. Toranzo, P.~S{\'a}nchez-Moreno, and J.~S. Dehesa.
\newblock Monotone measures of statistical complexity.
\newblock {\em Physics Letters A}, 380(3):377--380, 2016.

\bibitem{Zozor2017}
S.~Zozor, D.~Puertas-Centeno, and J.~S. Dehesa.
\newblock On generalized {S}tam inequalities and {F}isher--\text{R}{\'e}nyi
  complexity measures.
\newblock {\em Entropy}, 19(9):493, 2017.

\bibitem{Puertas2019}
D.~Puertas-Centeno.
\newblock Differential-escort transformations and the monotonicity of the
  {LMC}-\text{R}{\'e}nyi complexity measure.
\newblock {\em Physica A: Statistical Mechanics and its Applications},
  518:177--189, 2019.

\bibitem{IP2025}
R.~G. Iagar and D.~Puertas-Centeno.
\newblock A new pair of transformations and applications to generalized
  informational inequalities and {H}ausdorff moment problem.
\newblock {\em Communications in Nonlinear Science and Numerical Simulation},
  151:109091, 2025.

\bibitem{IP2025(b)}
R.~G. Iagar and D.~Puertas-Centeno.
\newblock Through and beyond moments, entropies and {F}isher information
  measures: new informational functionals and inequalities.
\newblock {\em arXiv preprint arXiv:2505.21015}, 2025.

\bibitem{Tsallis2009}
C.~Tsallis, A.~R.~Plastino R, and R.~F. Alvarez-Estrada.
\newblock Escort mean values and the characterization of power-law-decaying
  probability densities.
\newblock {\em Journal of Mathematical Physics}, 50(4), 2009.

\bibitem{Lindqvist1995}
P.~Lindqvist.
\newblock Some remarkable sine and cosine functions.
\newblock {\em Ricerche di Matematica}, 44:269, 1995.

\bibitem{Drabek1999}
P.~Dr{\'a}bek and R.~Man{\'a}sevich.
\newblock On the closed solution to some nonhomogeneous eigenvalue problems
  with {$p-$}{L}aplacian.
\newblock {\em Differential Integral Equations}, 12(6):773--788, 1999.

\bibitem{Shababi2016}
H.~Shababi, P.~Pedram, and W.~S. Chung.
\newblock On the quantum mechanical solutions with minimal length uncertainty.
\newblock {\em International Journal of Modern Physics A}, 31(18):1650101,
  2016.

\bibitem{Cveticanin2020a}
L.~Cveticanin, S.~Vujkov, and D.~Cveticanin.
\newblock Application of modified generalized trigonometric functions in
  identification of human tooth vibration properties.
\newblock {\em Communications in Nonlinear Science and Numerical Simulation},
  89:105290, 2020.

\bibitem{Gordoa2025}
P.~R. Gordoa, A.~Pickering, D.~Puertas-Centeno, and E.~V. Toranzo.
\newblock Generalized and new solutions of the {NRT} nonlinear
  {S}chr{\"o}dinger equation.
\newblock {\em Physica D: Nonlinear Phenomena}, 472:134515, 2025.

\bibitem{Bercher2012}
J.-F. Bercher.
\newblock On generalized {C}ram{\'e}r--{R}ao inequalities, generalized {F}isher
  information and characterizations of generalized {$q-$}{G}aussian
  distributions.
\newblock {\em Journal of Physics A: Mathematical and Theoretical},
  45(25):255303, 2012.

\bibitem{Lutwak2004}
E.~Lutwak, D.~Yang, and G.~Zhang.
\newblock Moment-entropy inequalities.
\newblock {\em The Annals of Probability}, 32(1B):757--774, 2004.

\bibitem{Chernoff1952}
H.~Chernoff.
\newblock A measure of asymptotic efficiency for tests of a hypothesis based on
  the sum of observations.
\newblock {\em The Annals of Mathematical Statistics}, pages 493--507, 1952.

\bibitem{Yamano2021}
T.~Yamano.
\newblock Skewed {J}ensen--{F}isher divergence and its bounds.
\newblock {\em Foundations}, 1(2):256--264, 2021.

\bibitem{Shelupsky1959}
D.~Shelupsky.
\newblock A generalization of the trigonometric functions.
\newblock {\em The American Mathematical Monthly}, 66(10):879--884, 1959.

\bibitem{Edmunds2012}
D.~E. Edmunds, P.~Gurka, and J.~Lang.
\newblock Properties of generalized trigonometric functions.
\newblock {\em Journal of Approximation Theory}, 164(1):47--56, 2012.

\bibitem{Bhayo2012}
B.~A. Bhayo and M.~Vuorinen.
\newblock On generalized trigonometric functions with two parameters.
\newblock {\em Journal of Approximation Theory}, 164(10):1415--1426, 2012.

\bibitem{Puertas2017}
D.~Puertas-Centeno, E.~V. Toranzo, and J.~S. Dehesa.
\newblock The biparametric {F}isher--\text{R}{\'e}nyi complexity measure and
  its application to the multidimensional blackbody radiation.
\newblock {\em Journal of Statistical Mechanics: Theory and Experiment},
  2017(4):043408, 2017.

\bibitem{Antolin2014}
J.~Antol{\'\i}n, J.~C. Angulo, S.~Mulas, and S.~L{\'o}pez-Rosa.
\newblock Relativistic global and local divergences in hydrogenic systems: A
  study in position and momentum spaces.
\newblock {\em Physical Review A}, 90(4):042511, 2014.

\bibitem{Martin2013}
A.~L. Mart{\'\i}n, J.~C. Angulo, and J.~Antol{\'\i}n.
\newblock Fisher-like atomic divergences: Mathematical grounds and physical
  applications.
\newblock {\em Physica A: Statistical Mechanics and its Applications},
  392(21):5552--5563, 2013.

\end{thebibliography}

\end{document}